\documentclass[pdflatex,sn-mathphys-ay]{sn-jnl}% Math and Physical Sciences Numbered Reference Style
\usepackage{graphicx}%
\usepackage{multirow}%
\usepackage{amsmath,amssymb,amsfonts}%
\usepackage{amsthm}%
\usepackage{mathrsfs}%
\usepackage[title]{appendix}%
\usepackage{xcolor}%
\usepackage{textcomp}%
\usepackage{manyfoot}%
\usepackage{booktabs}%
\usepackage{algorithm}%
\usepackage{algorithmicx}%
\usepackage{algpseudocode}%
\usepackage{listings}%
\usepackage{tikz}
\usetikzlibrary{arrows.meta}
\usetikzlibrary{calc}

\theoremstyle{thmstyleone}%
\newtheorem{theorem}{Theorem}%  meant for continuous numbers
\newtheorem{proposition}[theorem]{Proposition}% 

\theoremstyle{thmstyletwo}%

\theoremstyle{thmstylethree}%
\newtheorem{assumption}{Assumption}%

\newcommand{\bz}{{\textbf{z}}} 
\newcommand{\bgamma}{{\boldsymbol{\gamma}}} 
\newcommand{\bmu}{{\boldsymbol{\mu}}} 
\newcommand{\brho}{{\boldsymbol{\rho}}} 
\newcommand{\btheta}{{\boldsymbol{\theta}}} 
\newcommand{\blambda}{{\boldsymbol{\lambda}}} 
\newcommand{\bS}{{\textbf{S}}}

\begin{document}

\title[Article Title]{A latent space network model for dynamic neural latent embedding}

\author*[1,2,]{\fnm{Riccardo} \sur{Rastelli}}\email{riccardo.rastelli@ucd.ie}

\author[3]{\fnm{Shizhe} \sur{Chen}}\email{szdchen@ucdavis.edu}

\affil[1]{\orgdiv{School of Mathematics and Statistics}, \orgname{University College Dublin}, \country{Ireland}}

\affil[2]{\orgdiv{Rinn Artificial Intelligence}, \orgname{University College Dublin}, \country{Ireland}}

\affil[3]{\orgdiv{Department of Statistics}, \orgname{University of California, Davis}, \country{USA}}

\abstract{We introduce a novel latent space network model for analyzing multivariate time series of neural spike-train data. The methodology is motivated by an experimental study in mice, where neuronal responses were collected under a sequence of visual discrimination tasks. We adopt a latent variable framework to model the firing rates of aggregated brain areas, while simultaneously inferring the interactions between regions via a hidden network structure. This interaction network is embedded in a geometric latent space, enabling interpretable visualizations and novel model-based summaries. The proposed framework provides an intuitive interpretation of the latent variables, which bear a conceptual connection to node eigen-centrality measures. To capture temporal dependence, we incorporate a nested hidden Markov structure that can flexibly represent non-linear shifts that are induced by the changing of experimental conditions. We further establish theoretical properties of the model by deriving sufficient conditions that prevent degeneracy, thereby guiding our model assumptions. Overall, the proposed methodology provides a unified framework to characterize brain activity, its temporal dynamics, and spillover effects through a hidden latent space network model.}

\maketitle

%!TEX root = ../sn-article.tex

\section{Introduction}\label{sec:introduction}

Recent advances in neurophysiology have moved the field from single‐site recordings to probes that sample hundreds of neurons spread across multiple brain regions. Devices such as Neuropixels offer near-cellular spatial resolution and millisecond precision with minimal tissue disruption, allowing animals to perform sophisticated, hours-long tasks \citep{jun2017fully, steinmetz2018challenges, steinmetz2019distributed}. The technology advancement opens up opportunities to understand how neural systems evolve as behavior unfolds and learning progresses \citep{stringer2019spontaneous}. Consequently, the question of interest moves beyond identifying which brain regions activate to inferring dynamic shifts in the functional connectivity of neurons. In this paper, we provide a Bayesian framework that captures these evolving connectivity patterns using point processes and latent space modeling. 

Latent variable methods already play a prominent role in neural data analysis \citep{paninski2007statistical,cunningham2014dimensionality}. State-space models, for example, treat neural activity as governed by a low-dimensional latent process, known as the state, which is often modeled as dynamic systems or stochastic processes \citep{smith2003estimating, yu2008gaussian}. The neural activities are then linked to the latent state either through parametric models (e.g., point process) or nonparametric models using machine-learning methods. Although these techniques successfully summarize population activity, they typically represent the whole population with a single latent vector and thus do not yield an explicit, evolving network among nodes.  To bridge this gap, we develop a dynamic latent space network model in which each node, whether a single neuron or an anatomically defined area, occupies a coordinate in an evolving latent space; distances between coordinates govern the strength of directed interactions between nodes, allowing the inferred network to evolve smoothly as experimental conditions change. This construction retains the interpretability of geometric embeddings while directly targeting the neuroscientist's question of how functional connectivity reorganizes over the course of complex experiments.

Latent space network models~\citep{hoff2002latent}  have emerged as a prominent statistical framework for the analysis of networks, and, more generally, interaction data.
These models have been used in a variety of applied settings, including brain networks~\citep{durante2017nonparametric, wilson2020hierarchical, casarin2026bayesian}. For a recent review of latent space models and their applications, we refer the reader to \citet{kim2018review} and \citet{LPNM_2023}.
In our model, each node corresponds to a neuronal unit or an anatomically defined region. Directed edges encode the influence of one node's firing history on another's instantaneous rate, with edge strengths given by a function of the distance between the time-varying latent embeddings of the nodes.  The resulting log-rate vector satisfies a linear fixed-point equation that coincides with the stationary intensity of a multivariate Hawkes process~\citep{rousseau2018nonparametric}. Using a nested hidden Markov formulation, we characterize the temporal dynamics of the spikes both within and across trials, hence creating a formal framework that allows us to evaluate the persistence of activity patterns or their change over time.
Another central and novel aspect of our proposal regards the specification of the neuronal unit interactions, which is inspired by the literature on spatial hidden Markov models and multivariate time series analysis~\citep{bartolucci2022hidden, tancini2026spatio}. 
A key difference from these related research areas is that our formulation avoids model intractability as it relies on fixed point equations that can be solved analytically, linked to measures of node centralities, and, in particular to alpha-centrality~\citep{bonacich1987power,bonacich2001eigenvectorlike}. Remarkably, this connection is especially fitting for brain network analysis, where eigen-centrality measures have independently been highlighted as informative summaries of functional connectivity~\citep{bullmore2009complex,rubinov2010complex,lohmann2010eigenvector,mantzaris2013dynamic}.

As regards model fitting, we adopt a Bayesian framework and develop a suitable Markov chain Monte Carlo sampler to draw samples from the posterior distribution of the model. 
This connects to a vast literature on Bayesian latent space network models, which includes dynamic settings~\citep{sewell2015latent, loyal2024fast, casarin2026bayesian} and more recent time series frameworks that share similarities with our context~\citep{kaur2024latent, kaur2024dynamic}.
Although in this paper we focus on spike trains, our methodology extends to other relational event data such as social network interactions, contributing to the literature on network autocorrelation models \citep{sewell2017network, dittrich2017bayesian} and their extensions in other contexts \citep{sweet2020latent, tafakori2022measuring}. 

The rest of the paper is organized as follows. Section~\ref{sec:data} introduces the visual discrimination task and the neural recordings that motivate this analysis.
In Section \ref{sec:model}, we define our model and characterize its properties. In Section \ref{sec:methods} we outline our inferential framework, discussing the issue of non-identifiability commonly associated with latent space network models. In Section \ref{sec:simulation}, we propose two studies on synthetic data. In the first study, we show that the estimators that we use in our framework are consistent as the number of nodes or the number of trials increases. In the second study, we use a larger dataset, which is comparable to those analyzed in the real application of Section \ref{sec:rda}. In Section \ref{sec:discussion} we summarize and discuss the results of the research.

%!TEX root = ../sn-article.tex

\section{Visual Discrimination Task}\label{sec:data}

Visual discrimination tasks provide a controlled setting for studying how sensory evidence is transformed into choices and actions. In these tasks, head-fixed mice select the visual stimuli of higher contrasts by turning a wheel. Related wheel-based paradigms have been standardized across laboratories by the International Brain Laboratory \citep{ibl2021standardized}. The data analyzed here were collected by \citet{steinmetz2019distributed} and use a bilateral-stimulus variant of this general experimental paradigm. The original study involved 10 mice across 39 recording sessions, in each of which the animals completed several hundred trials of a two-alternative visual-decision task. Each trial begins with the simultaneous presentation of drifting-grating stimuli on two lateral screens (the \textit{stimulus onset}), with one of four contrast levels $\{0, 0.25, 0.5, 1\}$ on each side. After a brief delay, an auditory \textit{go cue} signals the mouse to indicate, via a forepaw-controlled wheel, which side carries the higher contrast; we refer to the moment of wheel-movement onset as the \textit{reaction time}. Correct choices are rewarded with water, whereas incorrect choices or failures to respond within one second are followed by negative feedback; we refer to the moment of feedback delivery as the \textit{feedback time}. Full details of the feedback protocol can be found in \cite{steinmetz2019distributed}. 

The four contrast levels on each screen define 16 bilateral stimulus combinations, which we group into four task classes: $\texttt{0-0}$, $\texttt{identical}$, $\texttt{0-other}$, and $\texttt{different}$. These classes summarize stimulus configurations that impose related decision requirements.
The $\texttt{0-0}$ class contains trials with zero contrast on both screens, in which the mouse is rewarded for keeping the wheel still for $1.5$ seconds. In the case of $\texttt{identical}$, the images are shown with some level of contrast, but they are indistinguishable. For this type of task, the mice do not have sufficient information to make a correct decision, and, according to the study design, a reward is randomly delivered if the mice turn the wheel in either direction \citep{steinmetz2019distributed}. The $\texttt{0-other}$ class contains trials with zero contrast on one screen and positive contrast on the other. The $\texttt{different}$ class contains trials with unequal positive contrasts, in which the mouse is expected to move the higher-contrast stimulus toward the center and receives positive feedback for a correct response.

During each session, spiking activity from hundreds of neurons in the \textit{left} hemisphere was recorded with two Neuropixels probes \citep{steinmetz2018challenges}. Probe positions were deliberately varied across sessions to sample a broad set of cortical and subcortical regions; in any given session, this yields a network of brain areas whose spiking activity unfolds in parallel over time, motivating a model that can capture directed interactions and their temporal dynamics within the session. In this analysis, we treat each anatomically identified brain area observed in a session as a network node. For each trial, we aggregate the spikes of all neurons assigned to the same brain area into 50-ms spike-count bins, producing the observed nodal time series. The real-data analysis focuses on Session 12, which contains 340 trials and 12 recorded area labels. We exclude the label \texttt{root} from the displayed analyses because it does not identify a unique anatomical area, leaving 11 brain areas.

We first examine whether the observed dependence among brain areas changes across behaviorally salient phases of the task and across the experimental session. Figure~\ref{fig:within_trial_correlations} compares empirical zero-lag correlation matrices from successful left-dominant trials in Session 12. The columns correspond to 100-ms windows following stimulus onset, reaction, and feedback, while the rows compare the first and last 20 eligible trials.
\begin{figure*}[ht]
    \centering
    \begin{minipage}[t]{0.33\textwidth}
        \textbf{A) initial trials stimuli}\par\vspace{0.15em}
        \includegraphics[width=\linewidth]{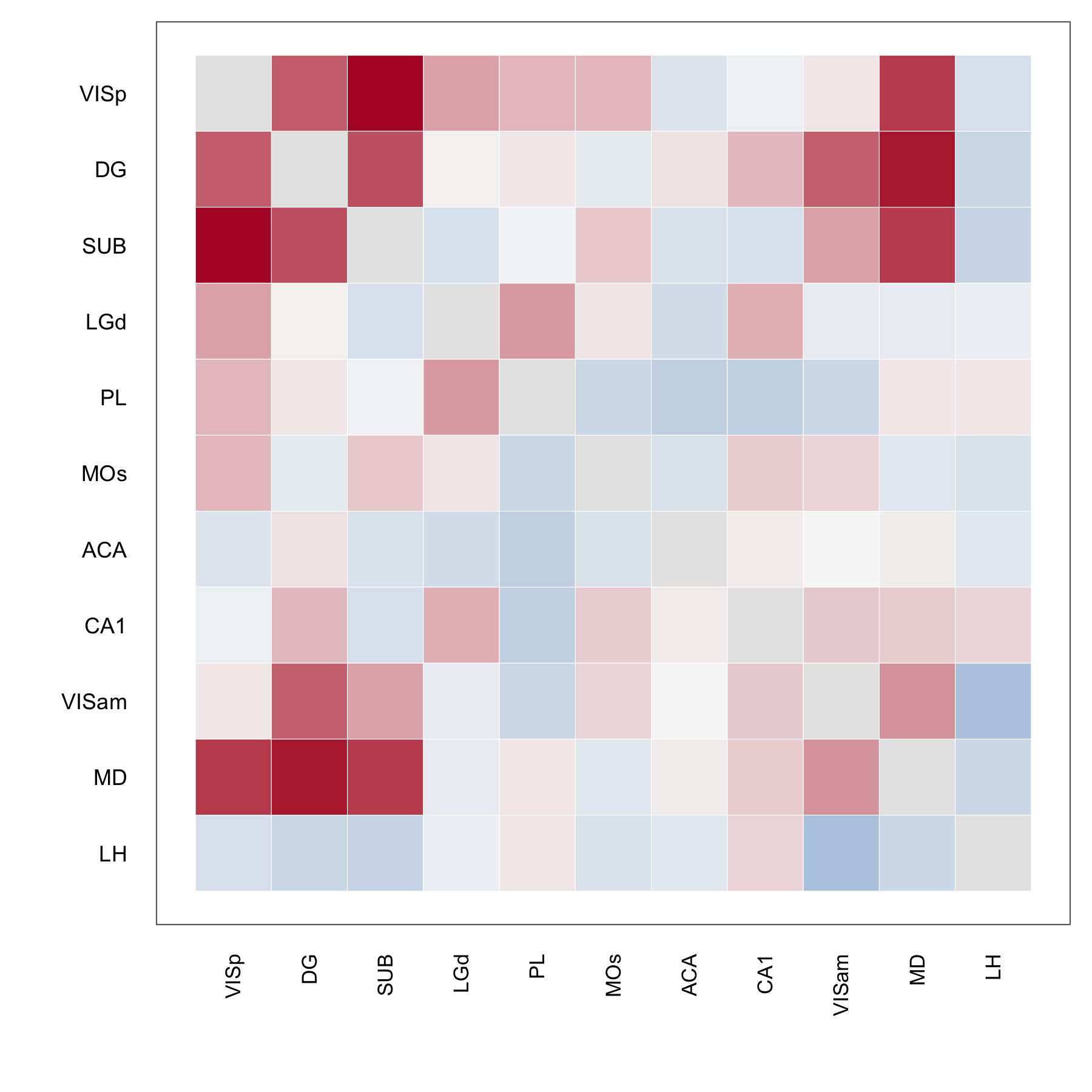}
    \end{minipage}\hfill
    \begin{minipage}[t]{0.33\textwidth}
        \textbf{B) initial trials reaction}\par\vspace{0.15em}
        \includegraphics[width=\linewidth]{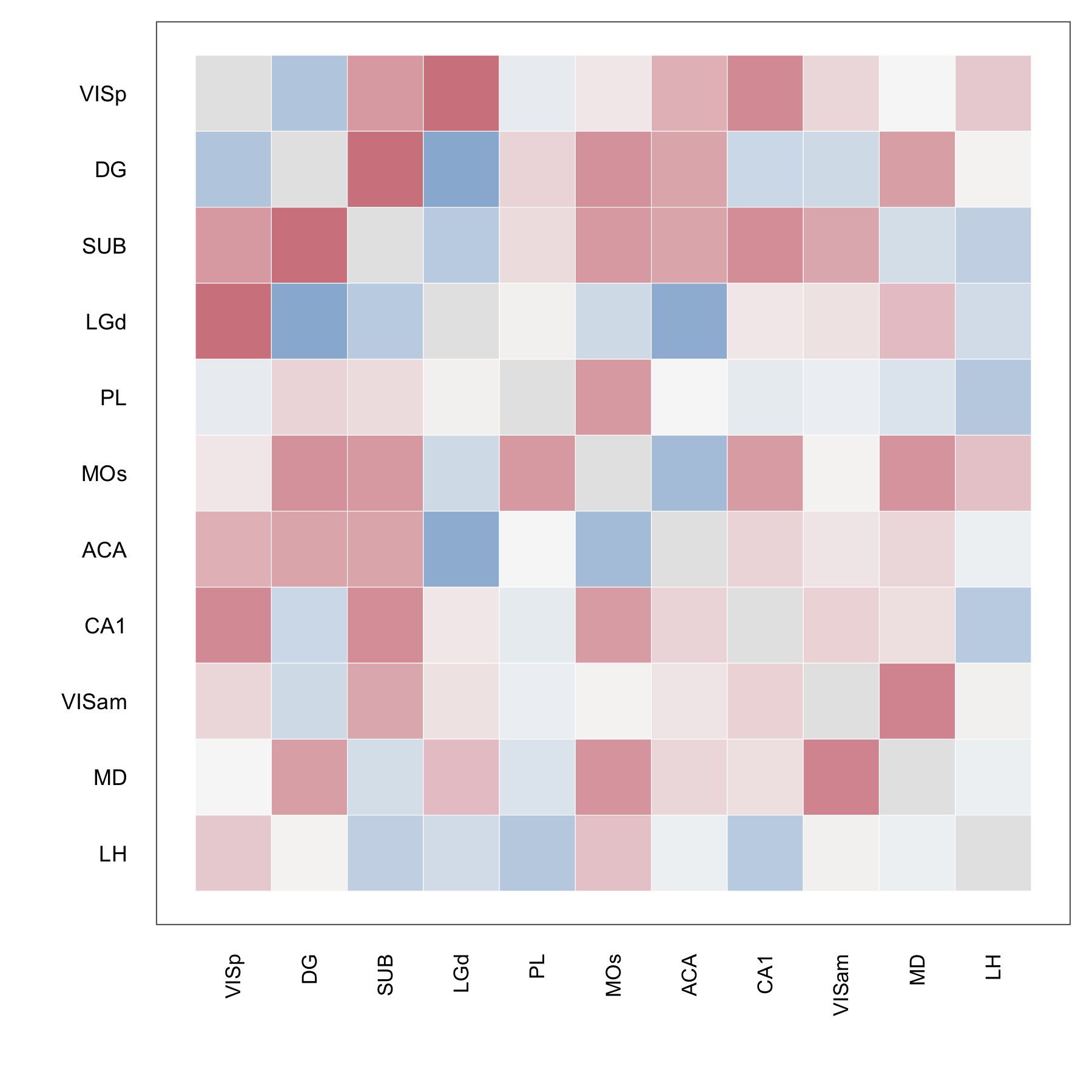}
    \end{minipage}\hfill
    \begin{minipage}[t]{0.33\textwidth}
        \textbf{C) initial trials feedback}\par\vspace{0.15em}
        \includegraphics[width=\linewidth]{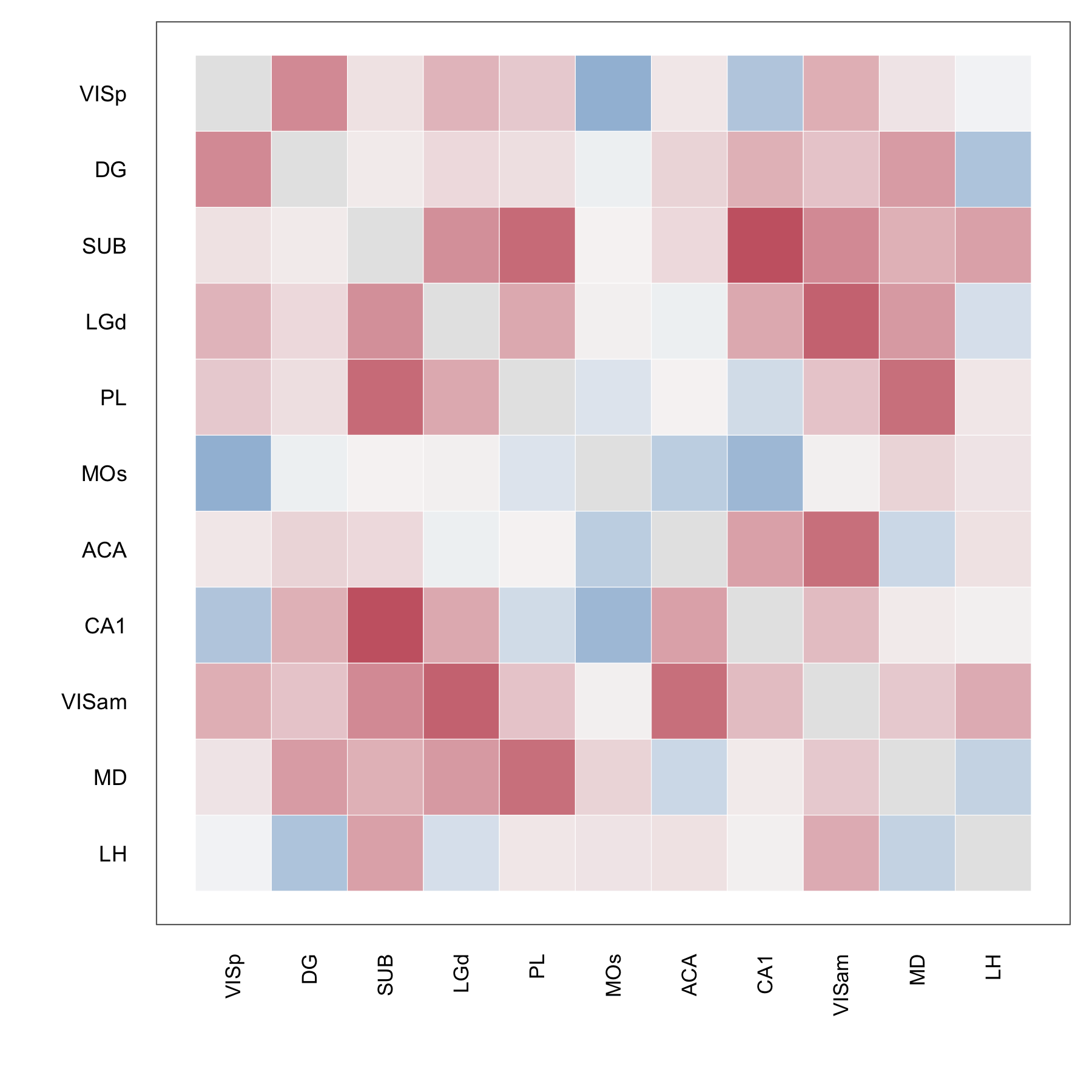}
    \end{minipage}
    \begin{minipage}[t]{0.33\textwidth}
        \textbf{D) final trials stimuli}\par\vspace{0.15em}
        \includegraphics[width=\linewidth]{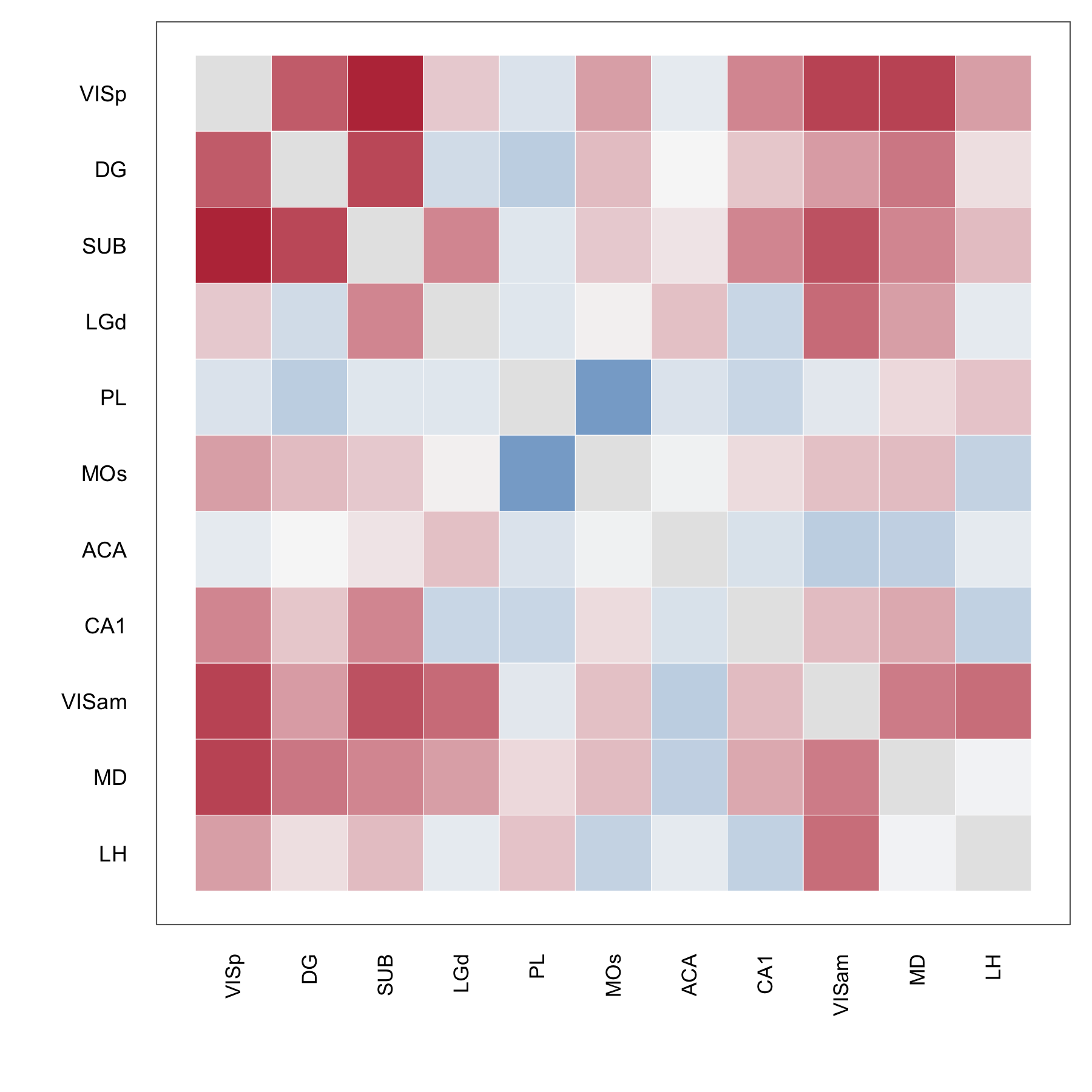}
    \end{minipage}\hfill
    \begin{minipage}[t]{0.33\textwidth}
        \textbf{E) final trials reaction}\par\vspace{0.15em}
        \includegraphics[width=\linewidth]{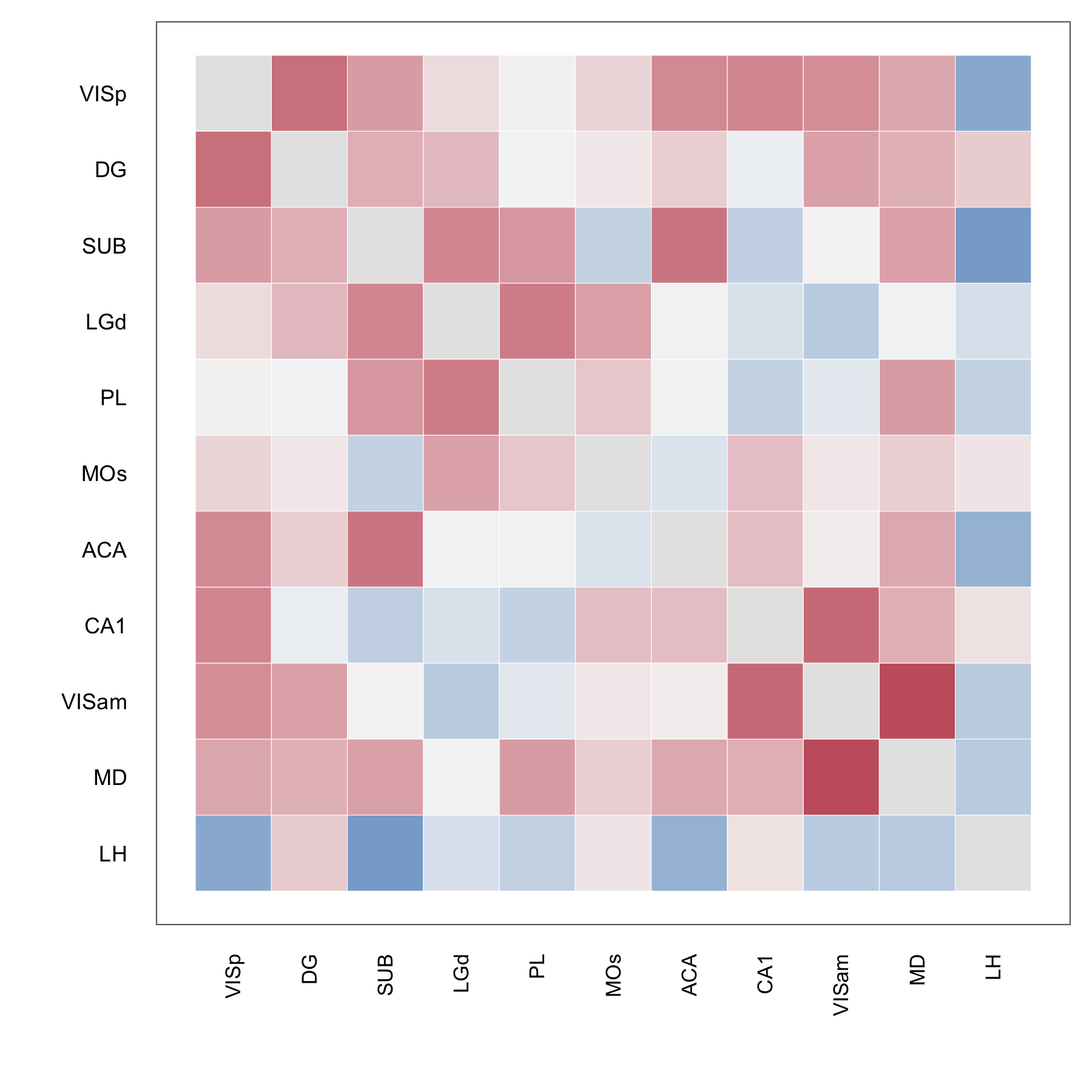}
    \end{minipage}\hfill
    \begin{minipage}[t]{0.33\textwidth}
        \textbf{F) final trials feedback}\par\vspace{0.15em}
        \includegraphics[width=\linewidth]{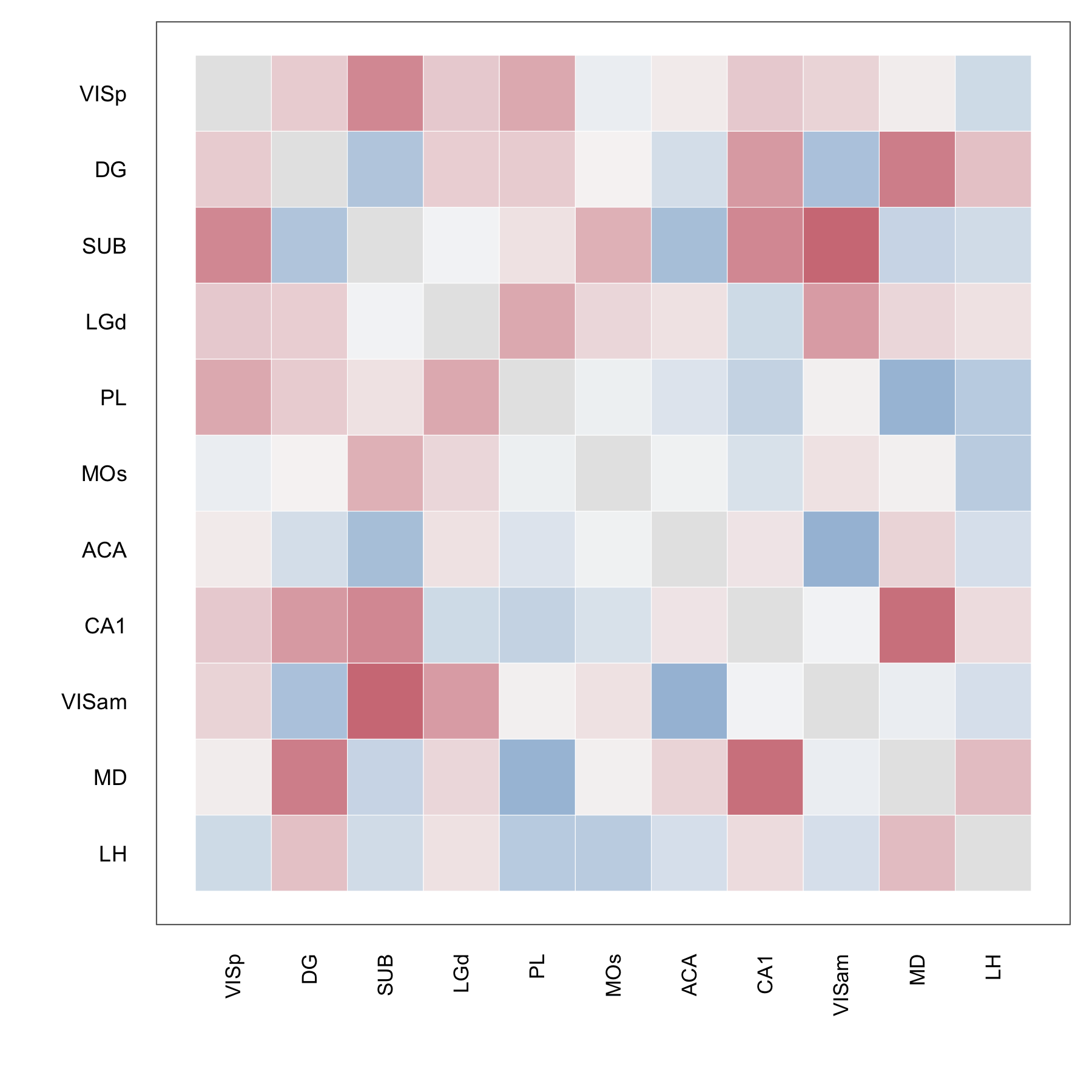}
    \end{minipage}
    \caption{\textbf{Empirical cross-area correlation across task phases.}
    The analysis includes successful left-dominant trials with left contrast equal to 1 and right contrast below 1. Three salient points are selected to separate the task phases, which are the stimuli onset, reaction time, and feedback time. Among the 42 trials with complete, non-overlapping 100-ms windows after all three salient events, panels A--C use the first 20 eligible trials and panels D--F use the last 20. Columns correspond to windows immediately after stimulus onset and before reaction, immediately after reaction and before feedback, and immediately after feedback delivery. Each window contains two 50-ms spike-count bins for each of 11 anatomical areas. Within each group and relative bin, area-specific means are removed, after which the 40 residual count vectors are pooled to calculate the correlation matrix. Blue cells indicate negative correlations, red cells indicate positive correlations, and gray diagonal cells are masked. All six panels use the same symmetric color scale (from $-0.7$ to $0.7$). The phase-specific patterns differ within both early and late trials; the early-to-late differences at a fixed phase are more modest. These matrices summarize empirical zero-lag co-activity and should not be interpreted as causal connectivity estimates.}
    \label{fig:within_trial_correlations}
\end{figure*}

Within both the early and late trial groups, the sign and magnitude of many pairwise correlations change across the stimulus-, reaction-, and feedback-aligned windows. Differences between the early and late panels at the same phase are more modest, but they suggest additional slow variation over the session. Such variation could reflect changes in engagement or fatigue, although these descriptive heatmaps do not identify the underlying behavioral mechanism. The phase-specific differences provide direct motivation for allowing the interaction structure to evolve within a trial.

\section{Model}\label{sec:model}

\subsection{Definition}

The observed data consist of $M$ trials, run sequentially on the same mouse and with independent experimental conditions. In each trial $m=1,\dots,M$, we record spike trains from $N$ brain regions over a time interval consisting of $T_m$ time bins of equal time width $\Delta > 0$. The spiking activity is represented as a collection of counts, for $m=1,\ldots, M$, $i=1,\ldots, N$, $t=1,\ldots, T_m$,
\[
y_{m,t,i} \;=\; \#\{\, \tau_{m,i,k} \in ((t-1)\Delta, t\Delta] \,\}, 
\]  
where  $\tau_{m,i,k}$ is the spike time of the $k$th spike from the $i$th region in the $m$th trial. In the motivating dataset, the number of trials is $M = 340$, the number of brain regions is $N = 12$, the time bin has a width of $0.05$ seconds, and the number of time bins ranges from $18$ to $54$ across trials.

We assume that the spike counts are conditionally independent and follow a Poisson distribution, $\forall m = 1,\dots,M;\ \forall t=1,\dots,T_m;\ \forall i = 1,\dots, N$:
\begin{equation*}
    \left. y_{m,t,i} \mid  \gamma_{m,t,i} \right. \sim {\rm Poisson} \{ \exp(\gamma_{m,t,i}) \},
\end{equation*}
where $\gamma_{m,t,i}$ is the log-rate of each of the spike counts. 
Here, the log-rate $\gamma_{m,t,i}$ is  further characterized as 
\begin{equation}\label{eq:model_2}
    \gamma_{m,t,i} = \mu_{m,t,i} + \frac{1}{N}\sum_{j\neq i}^N S_{m,t,i,j}\gamma_{m,t,j},
\end{equation}
where $\mu_{m,t,i} \in \mathbb{R}$ is a constant, and $S_{m,t,i,j} \in \mathbb{R}$ is an interaction score between the nodes $i$ and $j$ at time $t$ for trial $m$.
In Eq.~\eqref{eq:model_2}, we assume that the activity of each brain region is governed by a spontaneous rate $\mu_{m,t,i}$ and a network-induced effect, which is constructed using a dyadic interaction parameter $S_{m,t,i,j}$ that describes how much each brain region affects the others.
Critically, this interaction effect is not delayed with respect to time, in that it is possible for two brain regions to mutually and simultaneously excite each other, or compete against one another. Although Eq.~\eqref{eq:model_2} includes only the concurrent effects from other regions, the framework could also be generalized to include time-lagged terms; however, we do not pursue this extension here.

Our framework shares similarities and analogies with the multivariate time series literature \citep{lutkepohl2005stablevar}. Eq.~\eqref{eq:model_2} is in a similar format to a multivariate vector autoregression model, whereby the status of a node is explained by the status of the other nodes. As this dependency does not include a time lag, a potentially intractable dependency structure arises, akin to those often observed in spatial hidden Markov models \citep{pasanen2026hidden,douwes2025three}.

In matrix notation, we can rewrite Eq.~\eqref{eq:model_2} as
\begin{equation}\label{eq:model_matrix_notation_1}
    \bgamma_{m,t} = \bmu_{m,t} + \frac{1}{N}\bS_{m,t}\bgamma_{m,t}
\end{equation}
where $\bgamma_{m,t} \in \mathbb{R}^N$, $\bmu_{m,t}\in \mathbb{R}^N$, and $\bS_{m,t} \in \mathbb{R}^{N \times N}$ for $m=1,\dots,M$ and $t=1,\dots,T_m$. 
We can rewrite Eq.~\eqref{eq:model_matrix_notation_1} as 
\begin{equation}\label{eq:model_matrix_notation_2}
    \bgamma_{m,t} = \left(\mathbb{I} - \frac{1}{N}\bS_{m,t}\right)^{-1}\bmu_{m,t}
\end{equation}
within the set of points where the matrix $\mathbb{I} - \bS_{m,t}/N$ can be inverted. Here, $\mathbb{I}$ indicates the identity matrix of size $N\times N$. The existence conditions for $\bgamma$ are discussed in more detail in Section~\ref{sec:assumptions_and_properties}.
From Eq.~\eqref{eq:model_matrix_notation_1}, we can see the log-rates of the Poisson distribution being defined as a transformation of the baseline intercepts, using a matrix $\left(\mathbb{I} - \bS_{m,t}/N\right)^{-1}$ which encodes the effects of the network interactions.\\

As an additional layer for our model structure, we further characterize the similarity score between nodes by introducing a latent space projection model \citep{hoff2002latent} to explain the \textit{hidden} network interactions. 
This novel hierarchical formulation is inspired by recent contributions in the context of multivariate time series analysis \citep{zhu2020gnar,casarin2026bayesian}, but also by the literature on spatial hidden Markov models \citep{pasanen2026hidden,douwes2025three}.
For a set of latent positions $\left\{\bz_{m,t,i}\in \mathbb{R}^2 |\   m = 1,\dots, M;\   i = 1,\dots, N;\   t = 1,\dots, T_m\right\}$, we set the similarity score as
\begin{equation}\label{eq:def_S_matrix}
    S_{m,t,i,j} = \begin{cases}
        0 & \mbox{if} \quad i = j \\
        \bz_{m,t,i} \cdot \bz_{m,t,j} & \mbox{otherwise};
    \end{cases}
\end{equation}
for all $i$, $j$, $t$, $m$. This specification allows us to express various types of interactions between the brain regions. Letting $\theta_{m,t,i,j} \in [0,2\pi]$ denote the angle created by the vectors $\bz_{m,t,i}$ and $\bz_{m,t,j}$, then the dot product can be expressed as:
\begin{equation}\label{eq:def_z}
\bz_{m,t,i} \cdot \bz_{m,t,j} = \|\bz_{m,t,i}\|\|\bz_{m,t,j}\|\cos(\theta_{m,t,i,j}),
\end{equation}
where $\|\cdot\|$ is the Euclidean norm of a vector. From Eq.~\eqref{eq:def_z}, we can deduce that two nodes pointing in the same direction will tend to mutually excite, whereas any two nodes pointing in opposite directions will tend to inhibit each other's activities. Moreover, two perpendicular directions grant a dot product of zero, meaning that the two corresponding brain regions have no influence on each other. Another consideration regards the distance of nodes from the center of the space: the further away a node is located, the more its interactions with other nodes will be strengthened, either positively or negatively.

With the above model structure, we characterize the activity of a brain region at each point in time, combining a baseline spontaneous rate $\bmu_{m,t}$, and a network effect which is determined by the concurrent states of the other brain regions $\bz_{m,t}$. Crucially, both the spontaneous rates and the network effects may be allowed to change over time and across trials. Using this framework, we are able to model the dynamics of the brain activity network during the experiment, disentangling the evolution of the spontaneous rate from the network-induced rates. This makes it possible for us to highlight temporal changes, systemic changes and pattern shifts throughout the experiment with a model-based approach.

\subsection{Assumptions and properties}\label{sec:assumptions_and_properties}
Differently from recent work at the interface of time series and networks \citep{zhu2020gnar, casarin2026bayesian}, our formulation does not rely on a standard vector autoregressive (VAR) structure \citep{lutkepohl2005stablevar}, since the interactions are not time-delayed. On the other hand, our proposed simultaneous interaction mechanism shares notable similarities with the network literature on centrality measures, and in particular with the notion of alpha centrality \citep{bonacich1987power,bonacich2001eigenvectorlike}. In alpha centrality, a node's centrality is defined through a fixed-point relation in which the node's score depends on the scores of its neighbors \citep{bonacich1987power}:
\begin{equation}\label{eq:alpha_centrality}
    \gamma_{i}' = \mu' + \alpha'\sum_{j=1}^N S'_{ij}\gamma'_{j},
\end{equation}
where $\gamma'_i$ denotes the \emph{centrality score} of node $i$, $S'$ is a network adjacency matrix, $\mu'$ is a baseline term, and $\alpha'$ controls the magnitude of the network effect. We emphasize that this equation is the \emph{definition} of alpha centrality from the literature on network centrality measures, and it is included here only to motivate our construction. To reflect this, in Eq.~\eqref{eq:alpha_centrality} we intentionally use a notation that is only slightly different from that of Eq.~\eqref{eq:model_2}, with the dash sign to separate this equation from our model's.
In fact, in Eq.~\eqref{eq:model_2}, we define an analogous relationship in which each brain region has a spontaneous activity rate $\mu_{m,t,i}$, corrected by a network effect that would correspond to the choice of $\alpha = \frac{1}{N}$.

More in general, eigenvector-based centrality measures are routinely used as a summary tool to determine the importance of nodes in a network. The importance scores generally reflect the ``hubness" of a node and its influence on other nodes. 
Such measures have also been used, empirically, in brain networks \citep{bullmore2009complex,rubinov2010complex}, validating their high relevance in this particular applied context. For example, eigenvector centrality has been used in neuroimaging to identify influential regions in functional connectivity graphs \citep{lohmann2010eigenvector}. These applied works strongly support our use of an eigenvector-based measure for node activity as a principled starting point for modeling network-induced effects. \\

We remark that alpha centrality is well defined only if $\alpha$ is smaller than the inverse spectral radius of $S$. We thus introduce the following assumption.
\begin{assumption}\label{asmp:alpha}
The latent positions $\bz$ are constrained within the unit circle, i.e., $\|\bz_{m,t,i}\| < 1$ for $m = 1,\dots, M,\   t = 1,\dots, T_m, \  i = 1,\dots, N$.
\end{assumption}

Under this assumption, we can state the following proposition, which ensures that the model is well-defined for any choice of the parameters.
\begin{proposition}\label{prop:gamma_existence}
    If Assumption~\ref{asmp:alpha} holds, the parameters $\bgamma_{m,t}$ in \eqref{eq:model_2} are finite for any $\bmu_{m,t} \in \mathbb{R}^N$ and any $\bz_{m,t} \in \mathbb{R}^{N\times 2}$, for all $ m = 1,\dots, M,\   t = 1,\dots, T_m,$ and $ i = 1,\dots, N$.
\end{proposition}
\begin{proof}
    From Eq.~\eqref{eq:model_matrix_notation_2}, the parameters $\bgamma_{m,t}$ are well defined when the matrix $\mathbb{I} - \bS_{m,t}/N$ is invertible. This holds if and only if $\det\left(\mathbb{I} -\bS_{mt}/N\right) \neq 0$. We note that
    \begin{equation}
        \det\left(\mathbb{I} - \frac{1}{N}\bS_{m,t}\right) 
        = N^{-N} \det\left(N\mathbb{I} - \bS_{m,t}\right) 
        = (-1)^N N^{-N} \det\left(\bS_{m,t} - N\mathbb{I}\right)
        = (-1)^N N^{-N}  \psi_{\bS_{m,t}}\left( N \right),
    \end{equation}
    where $\psi_{\bS_{m,t}}\left( N \right)$ is the characteristic polynomial of $\bS_{m,t}$, evaluated at $N$. Thus, the matrix is invertible if $N$ differs from every eigenvalue of $\bS_{mt}$. This can be guaranteed with the condition 
    \begin{equation}\label{eq:proof_3}
    \nu(\bS_{m,t}) < N
    \end{equation}
    where $\nu(\bS_{m,t})$ indicates the spectral radius of $\bS_{m,t}$, i.e. the largest absolute value of its eigenvalues.

    We use now the general result that $\nu(\bS_{m,t}) \leq \|\bS_{m,t}\|$ for any matrix norm $\|\cdot\|$. In particular, if we consider the Frobenius norm of $\bS_{mt}$, this satisfies:
    $$
    \|\bS_{m,t}\|_F = \sqrt{\sum_{i,j}|S_{m,t,i,j}|^2} \leq \sqrt{\sum_{i,j}|\bz_{m,t,i} \cdot \bz_{m,t,j}|^2} < \sqrt{\sum_{i,j}|1|^2} = N
    $$
    thanks to Assumption~1. As a consequence, the inequality of Eq.~\eqref{eq:proof_3} holds true and thus the matrix is invertible.
\end{proof}

\subsection{Prior distributions}\label{sec:prior}

The temporal framework of the dataset that we analyze has a dual aspect. On the one hand, we have a sequence of trials that are executed in sequence. Thus, the neuronal activity may exhibit temporal dynamics that are best interpreted by considering trajectories over trials. For example, the mice may start exhibiting signs of fatigue as the experiment progresses, so this pattern could be captured by temporal dynamics that compare the baseline neuronal activity at the start of each trial. On the other hand, the parameters of the model may clearly change also within each trial, in direct response to the stimuli, feedback, or for any other trial-specific reason.
In order to provide a model-based representation of both of these temporal dynamics, we disentangle the temporal evolution using two hidden Markov model frameworks: one for the changes within each trial, and one for the changes in the starting condition of each trial.

In practice, we assume random walk priors on all parameters, ensuring that the changes over time are controlled by suitable standard deviation parameters. An important advantage of this prior is that it easily allows one to control how persistent the parameters may be over time, thus permitting easier interpretations and visualizations of the results.
The diagram in Figure \ref{fig:prior_diagram_mu} exemplifies this nested random walk structure for the intercepts, noting that the latent positions also follow the same dependency structure.

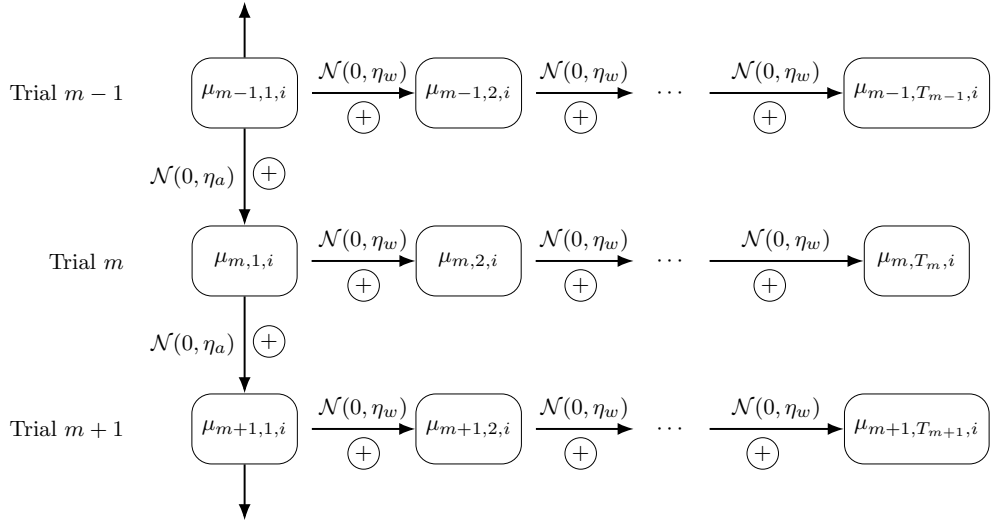
\begin{figure}[!h]
    \centering
    \resizebox{\textwidth}{!}{%
    \begin{tikzpicture}[
        font=\small,
        >=Latex,
        node distance=1.7cm,
        box/.style={draw, rounded corners=8pt, minimum height=10mm, minimum width=15mm, inner sep=4pt},
        plus/.style={draw, circle, inner sep=0pt, minimum size=4.5mm}
    ]
        % Row labels
        \node[anchor=east] (lab1) at (-1.6, 2.4) {Trial $m-1$};
        \node[anchor=east] (lab2) at (-1.6, 0.0) {Trial $m$};
        \node[anchor=east] (lab3) at (-1.6,-2.4) {Trial $m+1$};

        % --- Row m-1 ---
        \node[box] (a11) at (0, 2.4) {$\mu_{m-1,1,i}$};
        \node[plus] (p11) at (1.7, 2.05) {$+$};
        \node[box] (a12) at (3.2, 2.4) {$\mu_{m-1,2,i}$};
        \node[plus] (p12) at (4.8, 2.05) {$+$};
        \node (dots1) at (6.1, 2.4) {$\cdots$};
        \node[plus] (p1T) at (7.5, 2.05) {$+$};
        \node[box] (a1T) at (9.6, 2.4) {$\mu_{m-1,T_{m-1},i}$};

        \draw[->, thick] ($(a11.east)+(0.2,0)$) -- node[above] {$\mathcal{N}(0,\eta_w)$} (a12.west);
        \draw[->, thick] ($(a12.east)+(0.2,0)$) -- node[above] {$\mathcal{N}(0,\eta_w)$} ($(dots1.west)+(-0.2,0)$);
        \draw[->, thick] ($(dots1.east)+(0.2,0)$) -- node[above] {$\mathcal{N}(0,\eta_w)$} (a1T.west);

        % --- Row m ---
        \node[box] (b11) at (0, 0.0) {$\mu_{m,1,i}$};
        \node[plus] (q11) at (1.7,-0.35) {$+$};
        \node[box] (b12) at (3.2, 0.0) {$\mu_{m,2,i}$};
        \node[plus] (q12) at (4.8,-0.35) {$+$};
        \node (dots2) at (6.1, 0.0) {$\cdots$};
        \node[plus] (q1T) at (7.5,-0.35) {$+$};
        \node[box] (b1T) at (9.6, 0.0) {$\mu_{m,T_{m},i}$};

        \draw[->, thick] ($(b11.east)+(0.2,0)$) -- node[above] {$\mathcal{N}(0,\eta_w)$} (b12.west);
        \draw[->, thick] ($(b12.east)+(0.2,0)$) -- node[above] {$\mathcal{N}(0,\eta_w)$} ($(dots2.west)+(-0.2,0)$);
        \draw[->, thick] ($(dots2.east)+(0.2,0)$) -- node[above] {$\mathcal{N}(0,\eta_w)$} (b1T.west);

        % --- Row m+1 ---
        \node[box] (c11) at (0,-2.4) {$\mu_{m+1,1,i}$};
        \node[plus] (r11) at (1.7,-2.75) {$+$};
        \node[box] (c12) at (3.2,-2.4) {$\mu_{m+1,2,i}$};
        \node[plus] (r12) at (4.8,-2.75) {$+$};
        \node (dots3) at (6.1,-2.4) {$\cdots$};
        \node[plus] (r1T) at (7.4,-2.75) {$+$};
        \node[box] (c1T) at (9.6,-2.4) {$\mu_{m+1,T_{m+1},i}$};

        \draw[->, thick] ($(c11.east)+(0.2,0)$) -- node[above] {$\mathcal{N}(0,\eta_w)$} (c12.west);
        \draw[->, thick] ($(c12.east)+(0.2,0)$) -- node[above] {$\mathcal{N}(0,\eta_w)$} ($(dots3.west)+(-0.2,0)$);
        \draw[->, thick] ($(dots3.east)+(0.2,0)$) -- node[above] {$\mathcal{N}(0,\eta_w)$} (c1T.west);

        % Vertical links across trials (initial conditions)
        \node[plus] (v12) at (0.35, 1.25) {$+$};
        \node[plus] (v23) at (0.35,-1.15) {$+$};
        \draw[->, thick] (a11.south) -- node[left] {$\mathcal{N}(0,\eta_a)$} (b11.north);
        \draw[->, thick] (b11.south) -- node[left] {$\mathcal{N}(0,\eta_a)$} (c11.north);

        % Optional continuation arrows (top and bottom)
        \draw[->, thick] (a11.north) -- ++(0,0.8);
        \draw[->, thick] (c11.south) -- ++(0,-0.8);
    \end{tikzpicture}%
    }
    \caption{Nested hidden Markov prior structure on the parameters of the model. Each row corresponds to one trial; within a row, the boxes denote the intercepts at consecutive within-trial time points $t=1,\dots,T_m$ for brain area $i$, linked horizontally by Gaussian increments of variance $\eta_w$ from Eq.~\eqref{eq:mu_priors_1b}. The starting conditions $\mu_{m,1,i}$ of consecutive trials are linked vertically by Gaussian increments of variance $\eta_a$ from Eq.~\eqref{eq:mu_priors_1a}, so that slow dynamics propagate across trials. The same dependency structure is assumed for the latent positions $\bz_{m,t,i}$, with their own pair of within-trial and across-trial hyperparameters $(\zeta_w, \zeta_a)$ and $(\omega_w, \omega_a)$.}
    \label{fig:prior_diagram_mu}
\end{figure}

The random walk prior on the starting conditions of the trials implies that the intercept and the latent positions may not vary much across trials, reflecting the fact that the trials are repeated under similar conditions with a well-defined sequence of events and similar timings. However, this type of prior may also help in capturing a trend component, which could reflect fatigue of the mouse or an improvement in how the task is performed. 
At the end of each trial, a wash-out period with randomized duration was implemented to allow neurons to return to spontaneous states, where the randomized duration avoids anticipation. Our assumption of a Markovian sequence on initial conditions is thus consistent with the study design.

The random walk prior on within-trial parameter values is motivated by the presence of stimuli during the trial. It allows us to capture temporal changes associated with critical moments of the trial.\\

More formally, we define the priors for the intercept vectors $\bmu_{m,t} \in \mathbb{R}^N$ as, for $m=1,2,\ldots, M-1$,
\begin{equation}\label{eq:mu_priors_1a}
\begin{split}
    &\bmu_{m+1,1} \sim \mathcal{N}\left(\bmu_{m,1}, \eta_{a}\,\mathbb{I}_N\right)
\end{split}
\end{equation}
and for $m=1,2,\ldots, M$, $t=1,2,..., T_m-1$
\begin{equation}\label{eq:mu_priors_1b}
\begin{split}
    &\bmu_{m,t+1} \sim \mathcal{N}\left(\bmu_{m,t}, \eta_{w}\,\mathbb{I}_N\right),
\end{split}
\end{equation}
where $\eta_{a}$ and $\eta_{w}$ represent the variances for the increments of the intercepts across and within trials, respectively, and $\mathbb{I}_N$ denotes the $N\times N$ identity matrix. As concerns the latent positions, we use a parametrization in polar coordinates which allows us to easily characterize the constraint of Assumption~1:
\begin{equation*}
    \bz_{mti} = \frac{\exp(\rho_{m,t,i})}{\{1+\exp(\rho_{m,t,i})\}} \cdot \left(
    \begin{matrix}
        \cos(\theta_{m,t,i})\\ 
        \sin(\theta_{m,t,i})
    \end{matrix}\right)
\end{equation*}
where $\rho_{m,t,i} \in \mathbb{R}$ is the logit of the distance of a node to the center of the space, and $\theta_{m,t,i} \in [0,2\pi]$ is the angle that it forms with the horizontal axis, for every $m$, $t$ and $i$. Letting $\brho_{m,t} = (\rho_{m,t,1},\dots,\rho_{m,t,N})^\top$ and $\btheta_{m,t} = (\theta_{m,t,1},\dots,\theta_{m,t,N})^\top$, we specify, for $m=1,2,...,M-1$
\begin{equation}\label{eq:Z_prior_1a}
\begin{split}
    &\brho_{m+1,1} \sim \mathcal{N}\left(\brho_{m,1}, \zeta_{a}\,\mathbb{I}_N\right) \quad \forall m=1,\dots,M-1\\
    &\btheta_{m+1,1} \sim \mathcal{N}\left(\btheta_{m,1}, \omega_{a}\,\mathbb{I}_N\right) \quad \forall m=1,\dots,M-1\\
\end{split}
\end{equation}
and for $m=1,2,...,M-1, t=1,2,...,T_m-1$,
\begin{equation}\label{eq:Z_prior_1b}
\begin{split}
    &\brho_{m,t+1} \sim \mathcal{N}\left(\brho_{m,t}, \zeta_{w}\,\mathbb{I}_N\right)  \\
    &\btheta_{m,t+1} \sim \mathcal{N}\left(\btheta_{m,t}, \omega_{w}\,\mathbb{I}_N\right) \,.
\end{split}
\end{equation}
where $\zeta_{a}$ and $\omega_{a}$ indicate the variations across trials and $\zeta_{w}$ and $\omega_{w}$ indicate the variations within trials, respectively. The priors for the parameters indexed by $m=1$ and $t=1$ are chosen as improper uniform priors, e.g. densities that are perfectly flat over the whole real axis.

All of the variance parameters $\eta_{\cdot}$, $\zeta_{\cdot}$ and $\omega_{\cdot}$ are positive numbers that are chosen by the user. The rationale behind choosing these parameters is that high values make the model more flexible, because the parameters are able to move more over time, and so they can help in achieving a better model fit. By contrast, low values of the variance parameters lead to more clear interpretations and graphical visualizations. While this may sacrifice some of the flexibility of the model, it does promote parsimony and it makes it easier to describe any patterns that arise from the results. Within this context, we can see the variance parameters as regularization terms that let the user make a choice on the balance to strike between model fit and interpretability. 

An additional aspect to note here regards the identifiability of the model parameters. While the intercepts and the latent position parameters play very different roles in the model, there can still be near-non-identifiable configurations. For example, an increase in activity of a particular brain region may be explained both by an increasing intercept parameter, and by a latent position that tends to have positive interactions and gets pushed away from the center of the space. For these reasons, we do not recommend a framework where all parameters can vary across and within trials. Rather, we focus our attention on configurations where only one family of parameters can change across trials, and only one family of parameters can change within trials. As a technical note, the parameters that are not allowed to change are approximately fixed to the same value by choosing very small variances in their priors.

\section{Methods}\label{sec:methods}
\subsection{Markov chain Monte Carlo sampling}
We propose a parameter inference method based on Markov chain Monte Carlo sampling from the model's posterior distribution
\begin{equation}
    p\left(\bmu, \brho, \btheta \right) \propto \mathcal{L}_{\textbf{Y}}\left( \bmu, \brho, \btheta \right)\ p\left(\bmu \middle \vert \eta \right) \ p\left(\brho\middle \vert \zeta\right)\ p\left(\btheta\middle \vert \omega\right),
\end{equation}
where
\begin{equation}
    \mathcal{L}_{\textbf{Y}}\left( \bmu, \brho, \btheta \right) = \prod_{m,t,i} \left[\frac{\exp(-\lambda_{m,t,i}) \lambda_{m,t,i}^{y_{m,t,i}}}{y_{m,t,i}!}\right]
\end{equation}
is the Poisson likelihood function, and the remaining terms are priors
\begin{equation}
    \begin{split}
        &p\left(\bmu \middle \vert \eta \right) \propto \left\{\prod_{m=1}^{M-1} \prod_{i=1}^{N} \exp\left\{-\left(\mu_{m+1,1,i}-\mu_{m,1,i}\right)^2 / \eta_a\right\}\right]\left[\prod_{m=2}^{M} \prod_{t=1}^{T_m-1} \prod_{i=1}^{N} \exp\left\{-\left(\mu_{m,t+1,i}-\mu_{m,t,i}\right)^2 / \eta_w\right\}\right] \\
        &p\left(\brho \middle \vert \zeta \right) \propto \left[\prod_{m=1}^{M-1} \prod_{i=1}^{N} \exp\left\{-\left(\rho_{m+1,1,i}-\rho_{m,1,i}\right)^2 / \zeta_a\right\}\right]\left\{\prod_{m=2}^{M} \prod_{i=1}^{N} \prod_{t=1}^{T_m} \exp\left\{-\left(\rho_{m,t+1,i}-\rho_{m,t,i}\right)^2 / \zeta_w\right\}\right] \\
        &p\left(\btheta \middle \vert \omega \right) \propto \left[\prod_{m=1}^{M-1} \prod_{i=1}^{N} \exp\left\{-\left(\theta_{m+1,1,i}-\theta_{m,1,i}\right)^2 / \omega_a\right\}\right]\left[\prod_{m=2}^{M} \prod_{t=1}^{T_m} \prod_{i=1}^{N} \exp\left\{-\left(\theta_{m,t+1,i}-\theta_{m,t,i}\right)^2 / \omega_w\right\}\right] .
    \end{split}
\end{equation}

For a given configuration of model parameters $\left( \bmu, \brho, \btheta \right)$, the likelihood function is evaluated by first calculating the matrices $\bS_{m,t}$ for all $m$ and $t$, and then, using each of those matrices, the log-rates are calculated by iterating the fixed point equation as follows
\begin{equation}\label{eq:fixed_point_1}
    \gamma_{m,t,i} = \mu_{m,t,i} + \frac{1}{N}\sum_{j=1}^N S_{m,t,i,j}\gamma_{m,t,j},
\end{equation}
where we loop over $i=1,\dots,N$ repeatedly. In analogy to the so-called power method for the calculation of the eigenvalues of a matrix, this iterative procedure converges with few iterations (we set $100$ iterations as a conservative default), and its actual computing time is negligible with respect to the rest of the inferential procedure. Ultimately, this iterative step is numerically equivalent to the calculation of Eq.~\eqref{eq:model_matrix_notation_2}, which may in fact also be used as a (less computationally efficient) alternative. 

During the Markov chain Monte Carlo sampling, each individual parameter is updated in turn, using a Metropolis-Hastings step. However, another computational argument must be made here. Consider a generic latent position $\bz_{m,t,i}$, for some $m$, $t$ and $i$. Whenever the value of $\bz_{m,t,i}$ changes, this will affect all the values $\gamma_{m,t,\tilde{i}}$ for all $\tilde{i} = 1, \dots, N$, since all of the interactions of this node are used in the matrix inversion of Eq.~\eqref{eq:model_matrix_notation_2}. This implies that, even when updating a single latent position $\bz_{m,t,i}$, evaluating the resulting changes in   $\bgamma$ has quadratic computational cost. If we repeat this process for each parameter update, we end up with a prohibitive computational complexity. Hence, during the update of a generic position $\bz_{m,t,i}$, we evaluate all of the changes on the relevant matrix $\bS_{m,t}$, but then we only calculate Eq.~\eqref{eq:fixed_point_1} once for node $i$, with the rest of $\bgamma$ remaining unchanged. At the end of each loop of the sampler over all model parameters, we proceed with a full likelihood calculation to reset any error that might have accrued. While this creates an approximation in our procedure, we can monitor the error introduced at the end of each iteration through the full likelihood calculation. In practice, we find that the error introduced is immaterial, as the step size for the updates is minimal and thus any bias it creates appears completely negligible.  

\subsection{Post-processing}\label{sec:postprocess}
Since our framework relies on a latent projection model structure, the parameters are identifiable only up to rotations and reflections of the latent space. For completeness, we report this more formally in this section.

Without loss of generality, we focus on the common setting where nodes' positions in the latent space are all pointing in different directions. In other words, no two latent positions are collinear with the origin. This is essential to ensure a property of \textit{general position} of the matrix of latent positions, which justifies the following assumption.
\begin{assumption}
\label{asmp:general_position}
For every trial $m$ and time $t$, the latent positions $Z_{m,t}\in\mathbb{R}^{N\times 2}$ satisfy the following property: 
if $Z_{m,t},Z_{m,t}'\in\mathbb{R}^{N\times 2}$ both have rank $2$ and
\[
(Z_{m,t}Z_{m,t}^\top)_{ij}=(Z_{m,t}'(Z_{m,t}')^\top)_{ij},
\qquad \forall i,\ j \mbox{ such that } i\neq j,
\]
then
\[
Z_{m,t}Z_{m,t}^\top=Z_{m,t}'(Z_{m,t}')^\top.
\]
\end{assumption}
The assumption states that the off-diagonal terms of the Gram matrix are sufficient to determine the full Gram matrix. The assumption is motivated as its statement is a general known result for Gram matrices that follows from the onset described above. Then, we can proceed to state a proposition that formalizes the issue of non-identifiability.

\begin{proposition}\label{prop:identifiability}
Under Assumption~\ref{asmp:alpha} and Assumption~\ref{asmp:general_position}, let $Z_{m,t}, Z'_{m,t} \in \mathbb{R}^{N \times 2}$ be two configurations of latent positions stacked by row, both of rank $2$, for trial $m$ and time $t$. Let $\bS_{m,t}$ and $\bS'_{m,t}$ be the corresponding network interaction matrices, defined via Eq.~\eqref{eq:def_S_matrix}. Then $\bS_{m,t} = \bS'_{m,t}$ if and only if there exists a matrix $Q$ in the orthogonal group $O(2) = \{Q \in \mathbb{R}^{2 \times 2} : Q^\top Q = I_2\}$ such that $\bz'_{m,t,i} = Q\,\bz_{m,t,i}$ for every $i = 1, \dots, N$.
\end{proposition}

\begin{proof}
In this proof, we drop the fixed subscripts $(m,t)$ for ease of notation. We define $G = ZZ^\top$ and $G' = Z'(Z')^\top$. From the definitions, the off-diagonal entries of $G$ and $G'$ coincide with those of $\bS$ and $\bS'$, whereas the diagonal entries are $G_{ii} = \|\bz_i\|^2$ and $G'_{ii} = \|\bz'_i\|^2$ and the diagonal entries of $\bS'$ and $\bS$ are all zeros by definition. 

\emph{Sufficiency.} If $\bz'_i = Q\bz_i$ for some $Q \in O(2)$, then, for every pair $(i, j)$, $(\bz'_i)^\top \bz'_j = \bz_i^\top Q^\top Q \bz_j = \bz_i^\top \bz_j$, so $G=G'$. Therefore, recalling the relationships between $G, G'$ and $\bS, \bS'$, we deduce that $\bS' = \bS$.

\emph{Necessity.} By the definition of $\bS$ and $\bS'$, the off-diagonal entries of $G$ and $G'$ coincide. Since $Z$ and $Z'$ have rank $2$, $G$ and $G'$ are positive semi-definite matrices and their off-diagonal entries uniquely determine the diagonals, thanks to Assumption~\ref{asmp:general_position}. Hence $G=G'$, and thus $ZZ^\top=Z'(Z')^\top$. Since $Z$ and $Z'$ both have full column rank, the standard uniqueness result for full-rank Gram factorizations implies that there exists a $Q \in O(2)$ such that $\bz'_i = Q \bz_i$.

\end{proof}

% \textcolor{red}{I don't know if this statement is true ``By the full-rank assumption, the off-diagonal entries of a rank-two positive semi-definite matrix uniquely determine its diagonal". Below is a counterexample from codex..}

% \paragraph{Counterexample to diagonal recovery under full rank.}
% Full rank of the interaction matrix is not sufficient to ensure that the
% off-diagonal entries of a rank-two positive semidefinite Gram matrix uniquely
% determine its diagonal. Consider the two latent configurations
% \[
% Z=
% \frac{1}{3}
% \begin{pmatrix}
% 1 & 0\\
% 1 & 1\\
% 1 & 0
% \end{pmatrix},
% \qquad
% Z'=
% \frac{1}{6}
% \begin{pmatrix}
% 4 & 0\\
% 1 & 1\\
% 1 & 3
% \end{pmatrix}.
% \]
% Their Gram matrices are
% \[
% ZZ^\top=
% \begin{pmatrix}
% \frac{1}{9} & \frac{1}{9} & \frac{1}{9}\\
% \frac{1}{9} & \frac{2}{9} & \frac{1}{9}\\
% \frac{1}{9} & \frac{1}{9} & \frac{1}{9}
% \end{pmatrix}
% \]
% and
% \[
% Z'(Z')^\top=
% \begin{pmatrix}
% \frac{4}{9} & \frac{1}{9} & \frac{1}{9}\\
% \frac{1}{9} & \frac{1}{18} & \frac{1}{9}\\
% \frac{1}{9} & \frac{1}{9} & \frac{5}{18}
% \end{pmatrix}.
% \]

% \textcolor{red}{End of example.}

As a clarification, we highlight that, unlike in a Euclidean-distance-based latent space model, a translation $\bz \mapsto \bz + \textbf{c}$ would not cause non-identifiability issues for our model. In fact, the dot product $(\bz_i + \textbf{c})^\top (\bz_j + \textbf{c}) = \bz_i^\top \bz_j + (\bz_i + \bz_j)^\top \textbf{c} + \|\textbf{c}\|^2$ depends on $\textbf{c}$, hence, a translation has a direct effect on the value of the model likelihood.

The identifiability issue outlined in Proposition \ref{prop:identifiability} may not be concerning in an optimization setting: multiple optimal solutions would lead to exactly the same interpretations and conclusions. However, in a Bayesian sampling setting, the posterior draws that we obtain for each of the latent positions can be difficult to summarize. During the sampling, the latent space may undergo rotations or reflections which we are not able to detect. Common model-based summaries, such as posterior means, will be directly affected by this, making the results difficult to interpret.
This non-identifiability is a common and well known issue for latent position models \citep{hoff2002latent, shortreed2006positional, LPNM_2023}, which fortunately has a rather elegant and effective solution. In order to make the samples for the latent spaces interpretable, we use Procrustes matching, whereby each latent space is rotated clockwise or anticlockwise until a best match is found with the maximum a posteriori configuration. 
Once all sampled latent spaces are rotated to best match the configuration, they can be considered comparable, and common summary tools, such as posterior means, can be used. 

We point out an additional technical note. In the projection model, the posterior distribution of a node may resemble an arc, since the behavior of the node is directly represented by the angles that it forms with the other nodes. In such cases, its posterior mean calculated on the cartesian coordinates would naturally tend to lean towards the center of the space. To avoid this type of unwanted behavior, we construct the posterior averages for the latent positions by using the posterior means of their polar coordinates: in this way, the estimate of a latent position is constructed using its average distance from the center and its average angle with respect to the $x$-axis. This guarantees that the distance from the center of the space is properly gauged.

\section{Simulations}\label{sec:simulation}
In order to validate our methodology, we propose two simulation studies on artificial data. The first simulation study considers a large number of simulated datasets that have been generated under increasing numbers of trials $M$,  time points $T_m$, and  nodes $N$. This first study aims at demonstrating that the accuracy of the results is generally good even for relatively small datasets, and that the recovery of the model structure tends to improve as the number of nodes increases, or as the number of trials and times increase. One drawback for this study is that it requires generating and fitting many datasets, which is very time consuming for an inferential procedure based on sampling. So, we only run it on small datasets. The second simulation complements the first study by considering a single artificial dataset which is much larger in size, and similar to the one used in the real application that we analyze in Section \ref{sec:rda}. The simulation study gives evidence that the results are reliable for larger datasets, and that the latent space is estimated accurately for datasets similar to the real ones that we study.

\subsection{Study 1}\label{sec:simulation1}
The number of nodes varies in the set $\{4,\ 8,\ 12,\ 16\}$ and the number of trials varies in the set $\{5,\ 10,\ 15,\ 20\}$. For each setting, the number of time points within each trial is always set equal to the number of trials, and therefore also varies over $\{5,\ 10,\ 15,\ 20\}$. Under each combination of number of nodes and number of trials, we generate simulated datasets, and fit our model.
The intercept parameters $\mu_{m,t,i}$ are chosen as equally spaced values from $1$ to $5$ for $m = t = 1$. The angle parameters $\theta_{m,t,i}$ are also chosen as equally spaced values from $0$ to $2\pi$ for $m = t = 1$. These initial values of the parameters are thus set using a regular pattern, which is expected to give reasonable datasets that present a good amount of heterogeneity in how the nodes behave. The decision of generating parameters using a regular pattern is essential for small datasets, so that for all choices of $N$ and $M$ we obtain similar model structures and thus we can compare results across the setups. In fact, one could choose completely random initial values, but this would lead to some datasets that are meaningless, in that nodes would not display a wide enough range of behaviors. In such case, one would need to average out the results over a large number of repeated experiments, which is computationally impractical. By imposing a predetermined range for the intercepts and the angles, we can ensure that the one dataset that we generate will always show a realistic scenario and a pre-determined sequence of nodes' behaviors.

As regards the evolutions of these parameters over time, we consider first-order autoregressive processes with a drift, for all of them. A generic parameter $\phi_{m,t,i}$ varies as follows:
$$
\phi_{m+1,t,i} = \frac{1}{2}\phi_{1,1,i} + \mathcal{N}\left(\frac{1}{2}\phi_{m,t,i}, 0.1^2\right),
$$
$$
\phi_{m,t+1,i} = \frac{1}{2}\phi_{1,1,i} + \mathcal{N}\left(\frac{1}{2}\phi_{m,t,i}, 0.1^2\right),
$$
across and within trials, respectively. Here, a generic notation $\phi$ is used in place of either parameter $\mu$ or $\theta$.
Motivated by a similar argument that we have highlighted for the regular patterns above, the idea behind the stationary autoregressive processes is that we want to ensure that model generates relevant datasets that do not degenerate over time.

The parameters $\rho_{m,t,i}$ are set to $0.5$ for all nodes, times, and trials. This choice ensures that the network effect remains relevant to the same extent in all the components of the datasets, so that the difficulty in recovering the latent space remains comparable throughout a dataset and across datasets.

We run our algorithm under each setting of $N$ and $M$. The hyperparameters $\eta, \zeta, \omega$ are all set equal to $0.1^2$, to impose a strong prior that reduces the changes of parameters across times. The Markov chain Monte Carlo samplers are run for $90000$ iterations, where the last $9000$ iterations are used to create the final posterior sample by retaining every $30$th sample. These settings are intentionally chosen very conservatively since convergence was not formally checked for every dataset.
The actual posterior sample size that we use to construct estimates for each parameter is equal to $300$.  

Figure \ref{fig:sim_1a_scatters_mu} represents the true values of $\mu_{m,t,i}$ against their corresponding estimates under each setting considered. 
\begin{figure}[htbp!]
\begin{center}
 \includegraphics[width=0.99\textwidth]{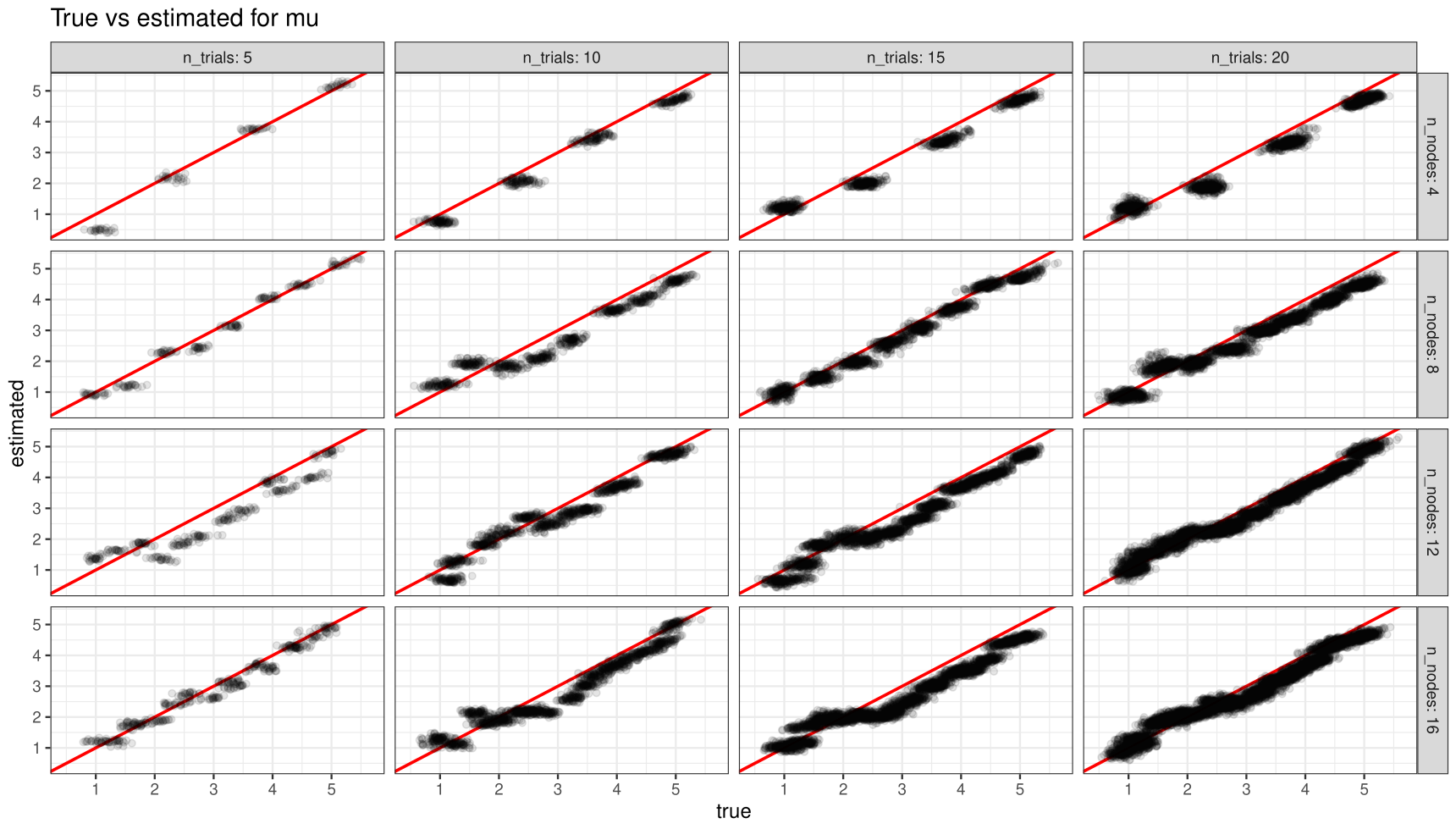}
 \caption{Simulation 1. Generated values of the intercept parameter against its estimated posterior mean, for each setting considered. Estimates are fairly accurate across the board, with minimal bias as the number of nodes and trials increase.}
 \label{fig:sim_1a_scatters_mu}
\end{center}
\end{figure}
The figure illustrates that the intercept parameters are generally well recovered in all settings, with results improving as the number of nodes and trials increase.
Figure \ref{fig:sim_1a_scatters_lambda} shows a very analogous result for the log-rates, highlighting an excellent model fit, overall. 
\begin{figure}[htbp!]
\begin{center}
 \includegraphics[width=0.99\textwidth]{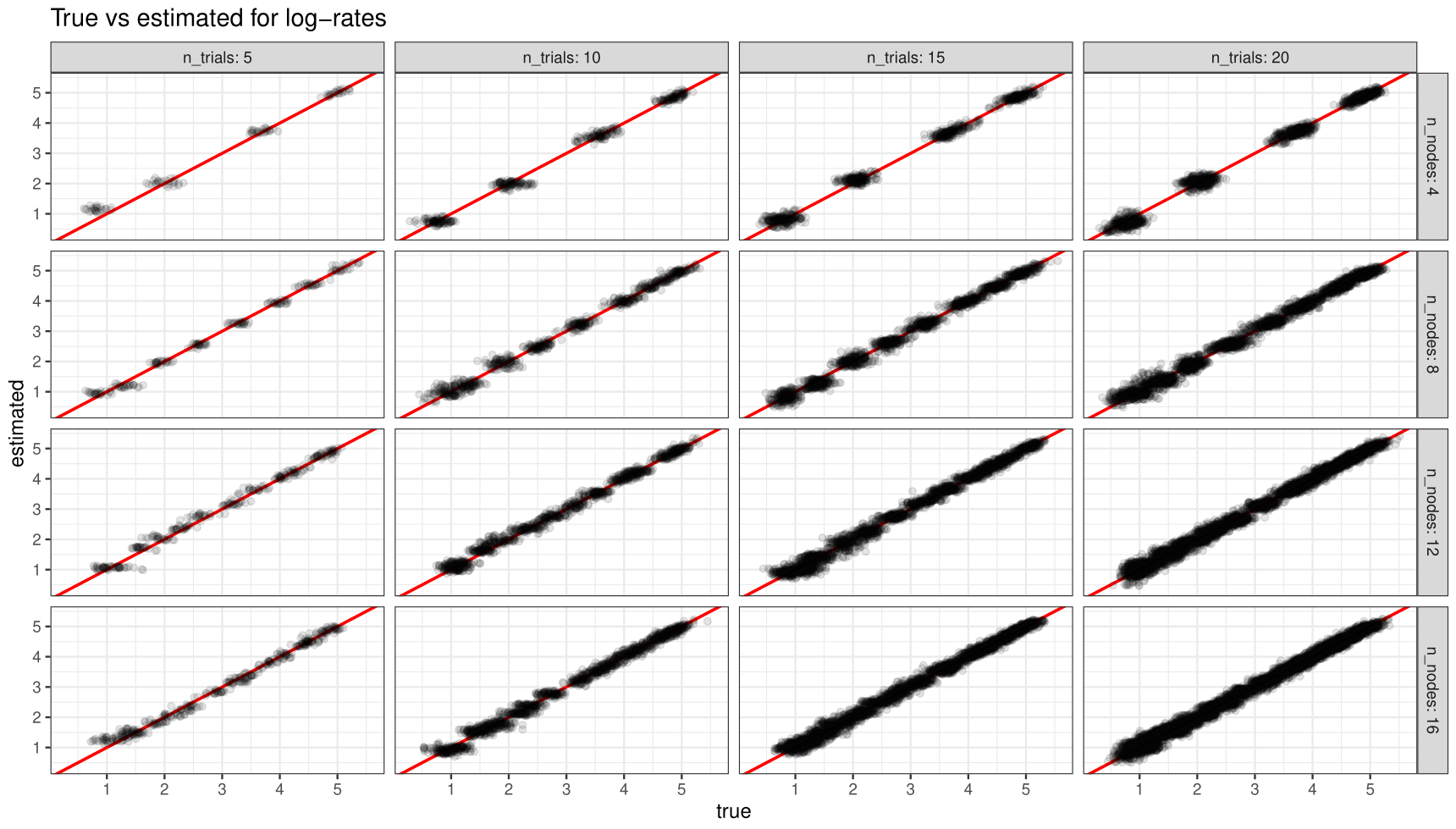}
 \caption{Simulation 1. Generated values of the log-rates against their estimated posterior mean, for each setting considered. Model fit is excellent for all datasets.}
 \label{fig:sim_1a_scatters_lambda}
\end{center}
\end{figure}
These two figures demonstrate that the inferential procedure is able to fit the data very well, and that the intercept parameters are accurately estimated. However, a critical element of our model is the latent space, which involves a substantial number of additional  parameters, and thus may be more challenging to estimate.
Figure \ref{fig:sim_1a_scatters_gamma} presents another figure in the same style for the interactions between nodes, as measured by their dot products.
\begin{figure}[htbp!]
\begin{center}
 \includegraphics[width=0.99\textwidth]{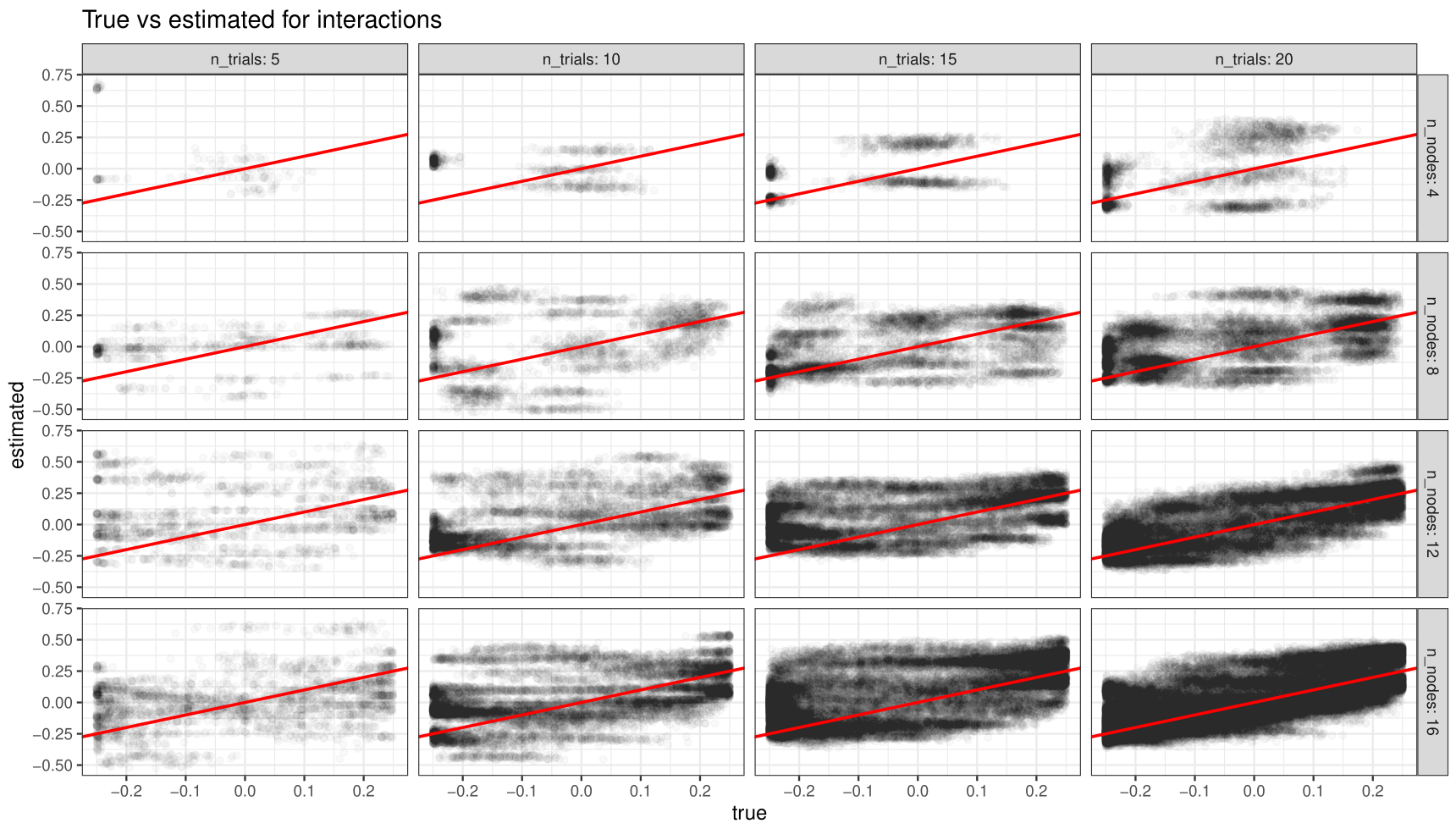}
 \caption{Simulation 1. Generated values of the dot products against their estimated posterior mean, for each setting considered. While the uncertainty remains large, the estimation of the network effects tends to improve as the number of nodes and trials increase.}
 \label{fig:sim_1a_scatters_gamma}
\end{center}
\end{figure}
In this case, the figure shows much larger errors. However, accuracy tends to improve as the numbers of nodes and trials increase.
The results indicate that the network effects are well captured by the model, even in the case of relatively small datasets.
While some errors can be present when recovering the latent space, these do not seem to lead to a worse fit of the model.
The figures in combination provide evidence that the method can successfully disentangle the baseline intercepts from the latent space network effect, estimating both parts of the model efficiently. 

% As mentioned, we also ran the experiment with all hyperparameters set to $0.1^2$, which is more representative of the value used to generate some of the parameters. In this case, the results for the baseline intercepts and log-rates were comparable to those of Figures \ref{fig:sim_1a_scatters_mu} and \ref{fig:sim_1a_scatters_lambda}, suggesting an excellent model fit (results not reported). However, inference on the latent space resulted much less convincing, as shown in Figure~\ref{fig:sim_1b_scatters_gamma}.
% \begin{figure}[htbp!]
% \begin{center}
%  \includegraphics[width=0.99\textwidth]{Figures/sim_1b_scatters_gamma.png}
%  \caption{Simulation 1. Generated values of the dot products against their estimated posterior mean, for each setting considered and hyperparameters equal to $0.1^2$. The latent space is not recovered well in this case.}
%  \label{fig:sim_1b_scatters_gamma}
% \end{center}
% \end{figure}
% This suggests that the model is overfitting, and that one could prevent this overparametrization issue by selecting more informative priors.

To conclude the first simulation study, we illustrate the computing times in Figure~\ref{fig:sim_1_ctime}.
\begin{figure}[htbp!]
\begin{center}
 \includegraphics[width=0.49\textwidth, page = 1]{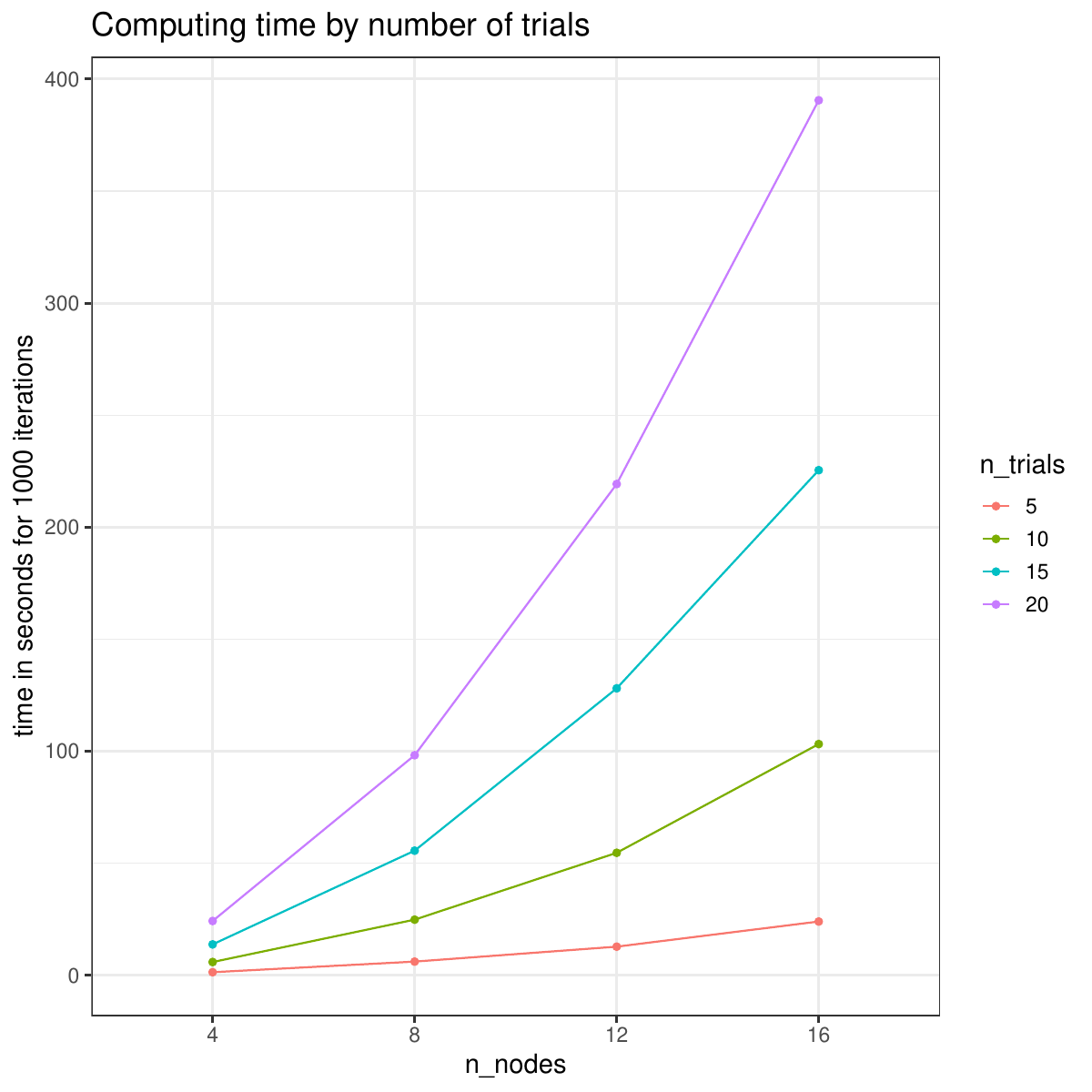}
 \includegraphics[width=0.49\textwidth, page = 2]{Figures/sim_1_ctime.pdf}
 \caption{Simulation 1. Computing time for all scenarios considered. As is common for latent space models, the computing cost grows quadratically with the number of nodes, shown in the left plot. The plot on the right shows instead a linear increase with respect to the number of time instances, defined as the number of trials times the number of time points in each trial (in simulation study 1, these two numbers are set to be equal to each other, in each dataset).}
 \label{fig:sim_1_ctime}
\end{center}
\end{figure}
The left panel highlights the quadratic increase in computing time as the number of nodes increases. The right panel shows instead the computing time as the number of nodes is fixed and the number of time instances varies. Here, the number of time instances on the x-axis is defined as the number of trials, multiplied by the number of time points in each trial, effectively representing the longitudinal size of the data.
This second plot shows that computing time increases linearly with the number of time instances, as expected for a dynamic latent space network model.

\subsection{Study 2}\label{sec:simulation2}
In the second simulation study we generate a single dataset with $N=20$, $M=50$ and $T_m = 25$ for all $m=1,\dots,M$. This setting is chosen since it is remarkably similar to the real datasets that we consider in this work. Unlike in the first simulation study, we generate the artificial data using our model's temporal dynamics. We generate the intercept at the very first time as $\mu_{1,1,\cdot} \sim \mathcal{N}(10,1)$ and then use the Markov properties across trials and across times to sample all other intercept values, as per Eq.~\eqref{eq:mu_priors_1a} and Eq.~\eqref{eq:mu_priors_1b}. The hyperparameters that we choose are $\eta_a = 0.05$ and $\eta_w = 0.005$. Similarly, for the latent positions, we use the temporal dynamics of Eq.~\eqref{eq:Z_prior_1a} and Eq.~\eqref{eq:Z_prior_1b} with hyperparameters $\zeta_a = \omega_a = 0.005$ and $\zeta_w = \omega_w = 0.05$. These values of the hyperparameters are identical to those that we use in our real data applications: more details on the motivation for these choices are given in Section \ref{sec:rda}.

We run our algorithm and, after burn-in and thinning, we obtain posterior samples of size $500$ for each of the model parameters. We confirmed satisfactory convergence of the MCMC chains using trace plots and convergence diagnostics\citep{raftery1992practical}. We then proceeded to compare our estimates with the true values of the parameters that have generated the data.

Figure \ref{fig:sim_2} illustrates the accuracy of the estimated parameters for the baseline intercept parameters, and for the network effects. 
\begin{figure}[htbp!]
\begin{center}
 \includegraphics[width=0.49\textwidth]{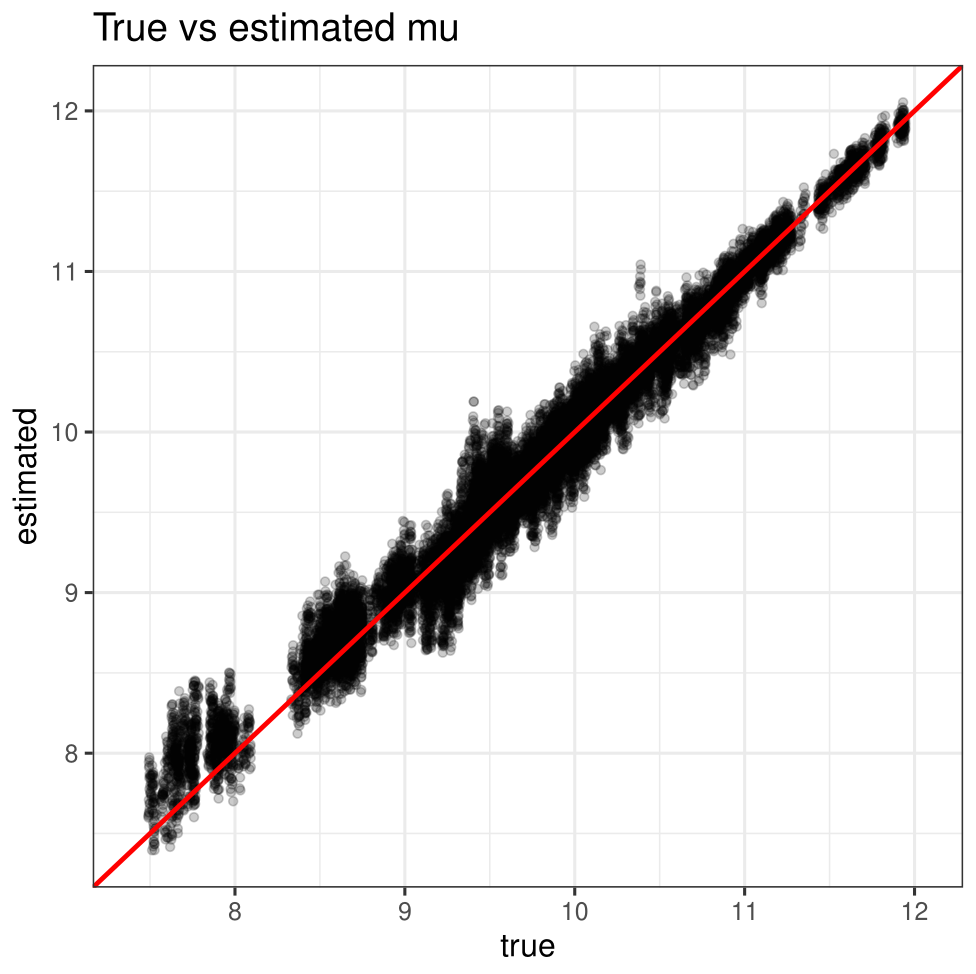}
 \includegraphics[width=0.49\textwidth]{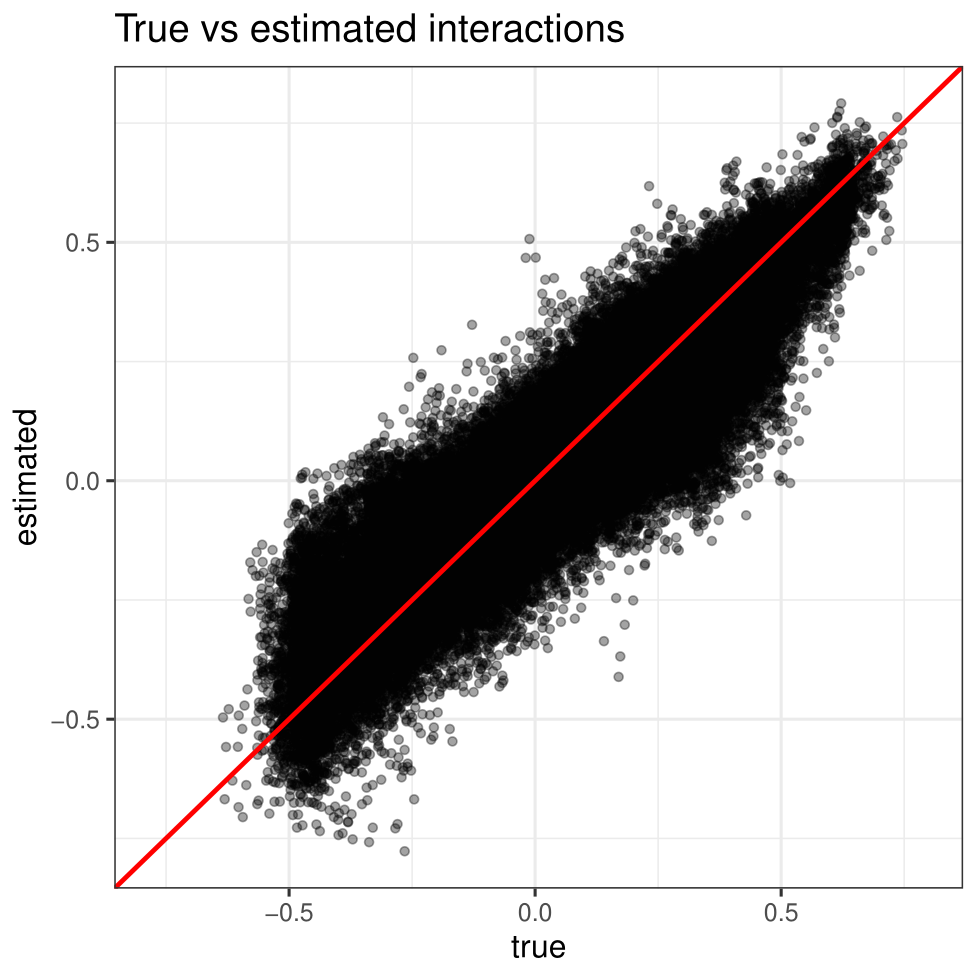}
 \caption{Simulation 2. The left panel shows the true generated values of the baseline parameter against its estimated posterior mean. The right panel shows the true generated values of all the interaction parameters (the dot products) against their estimated posterior mean.}
 \label{fig:sim_2}
\end{center}
\end{figure}
The results highlight that both the intercepts and all the pairwise network effects are very well recovered.
The magnitude of the errors for $\mu$ tends to decrease as the value of $\mu$ increases: this can be explained by the fact that, under a Poisson distribution, large observed counts can give more precise information on the exact value of the corresponding log-rates.

After the post-processing step described in Section \ref{sec:postprocess} using the true positions as reference, the aggregate latent space is obtained by calculating the posterior mean of each latent position for all nodes and times, and then averaging out these positions across time.   
Figure \ref{fig:sim_2b} illustrates the true and fitted latent spaces.
\begin{figure}[htbp!]
\begin{center}
 \includegraphics[width=0.6\textwidth]{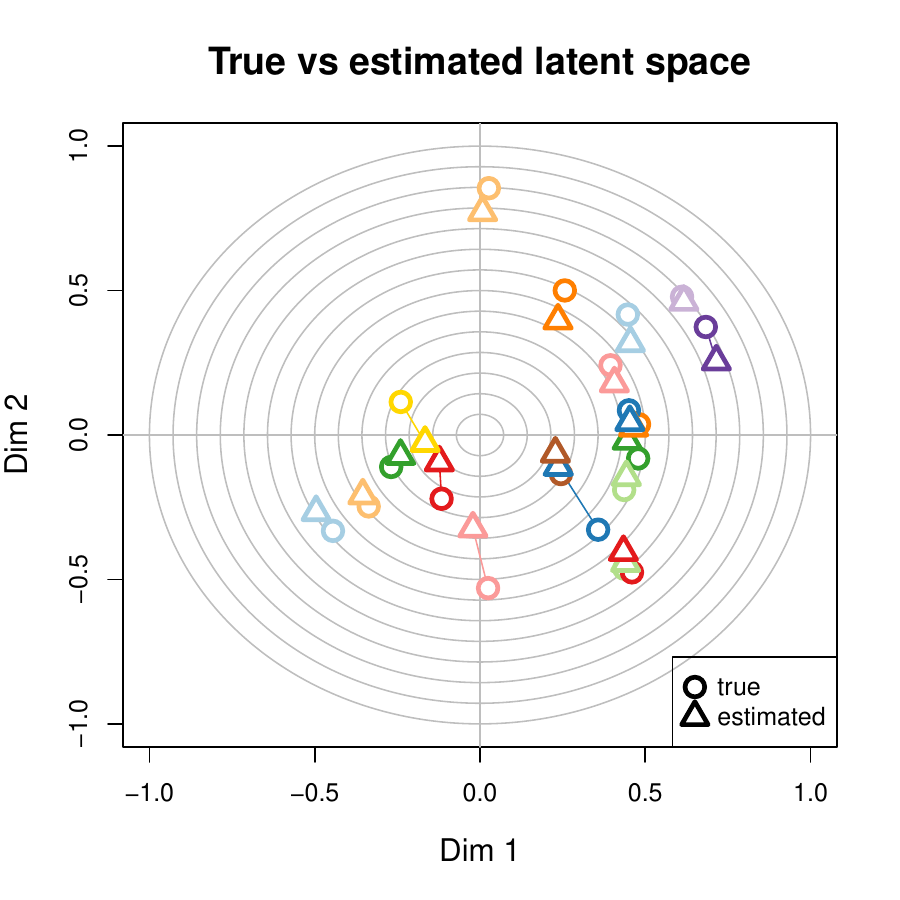}
 \caption{Simulation 2. True and estimated latent positions for the nodes, averaged out over times and trials.}
 \label{fig:sim_2b}
\end{center}
\end{figure}
The results highlight relatively small estimation errors which do not have any material effect on the interpretation of the results. 
Overall, the results indicate that, for a dataset similar to our real data applications, the proposed method can accurately estimate the relevant model parameters, and, in particular, the latent space that generates the network interactions.

%!TEX root = ../sn-article.tex

\section{Real Data Analysis}\label{sec:rda}

\subsection{Exploratory data analysis}\label{sec:eda}

We consider the recording of Session 12 from \cite{steinmetz2019distributed} that contains $M=340$ trials on Mouse Lederberg. The spikes are recorded from $12$ brain regions; however, we remove the \texttt{root} area from our illustrations of the results, since it is a collection of neurons with no known regions.  Figure~\ref{fig:eda_raw_data} shows the spiking activity of the recorded brain regions over the course of the session. The raw activity does not, on inspection, suggest any single dominant pattern that would be captured by a low-dimensional summary. 
\begin{figure}[htbp!]
\begin{center}
 \includegraphics[width=1\textwidth]{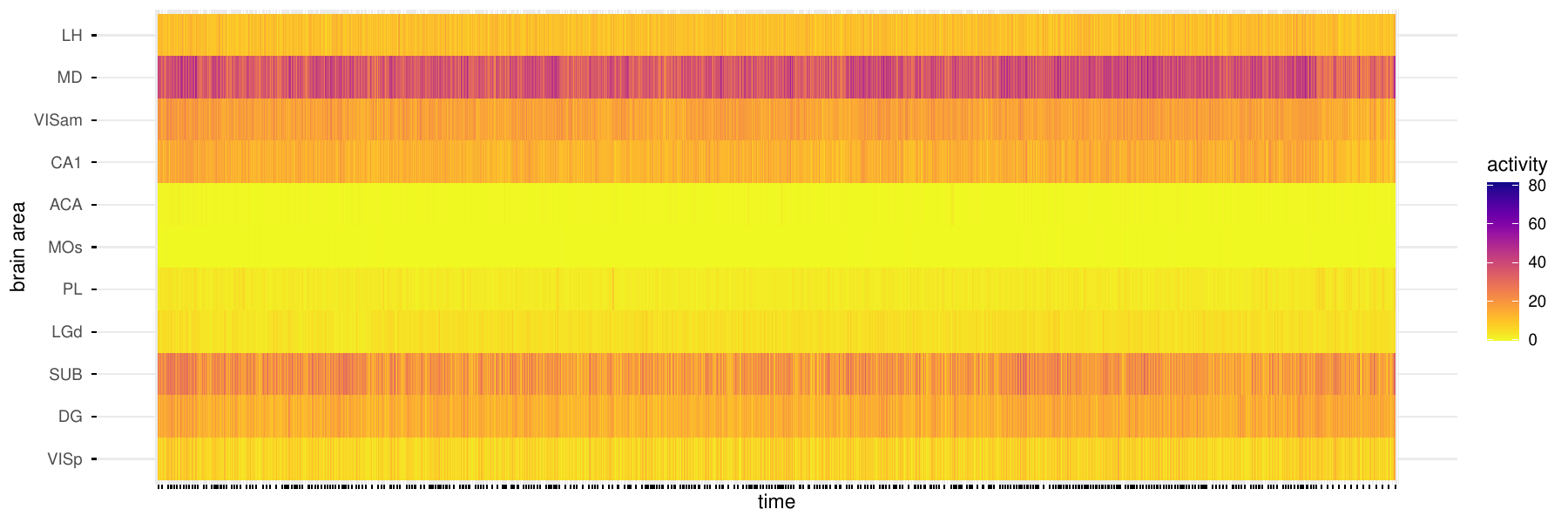}
 \caption{Activity level of $11$ brain areas from the recording of Session $12$. The session is composed of $340$ trials which are illustrated side-by-side, with ticks on the x-axis indicating their respective start times. The data indicates the spiking activity from the neuron associated to each brain area, over time intervals of $0.05$ seconds.}
 \label{fig:eda_raw_data}
\end{center}
\end{figure}

The experimental events and response rules are described in Section~\ref{sec:data}. For the analyses below, trials are aligned to stimulus onset, reaction time, or feedback time; the analysis window also includes the 0.2 seconds preceding stimulus onset to represent baseline neural activity.
Figure~\ref{fig:eda_salient} illustrates the distributions of these inter-event intervals across trials. The timing of salient events varies substantially across trials, motivating an analysis framework that aligns trials to behaviorally meaningful moments rather than to absolute time.

\begin{figure}[htbp!]
\begin{center}
 \includegraphics[width=0.9\textwidth]{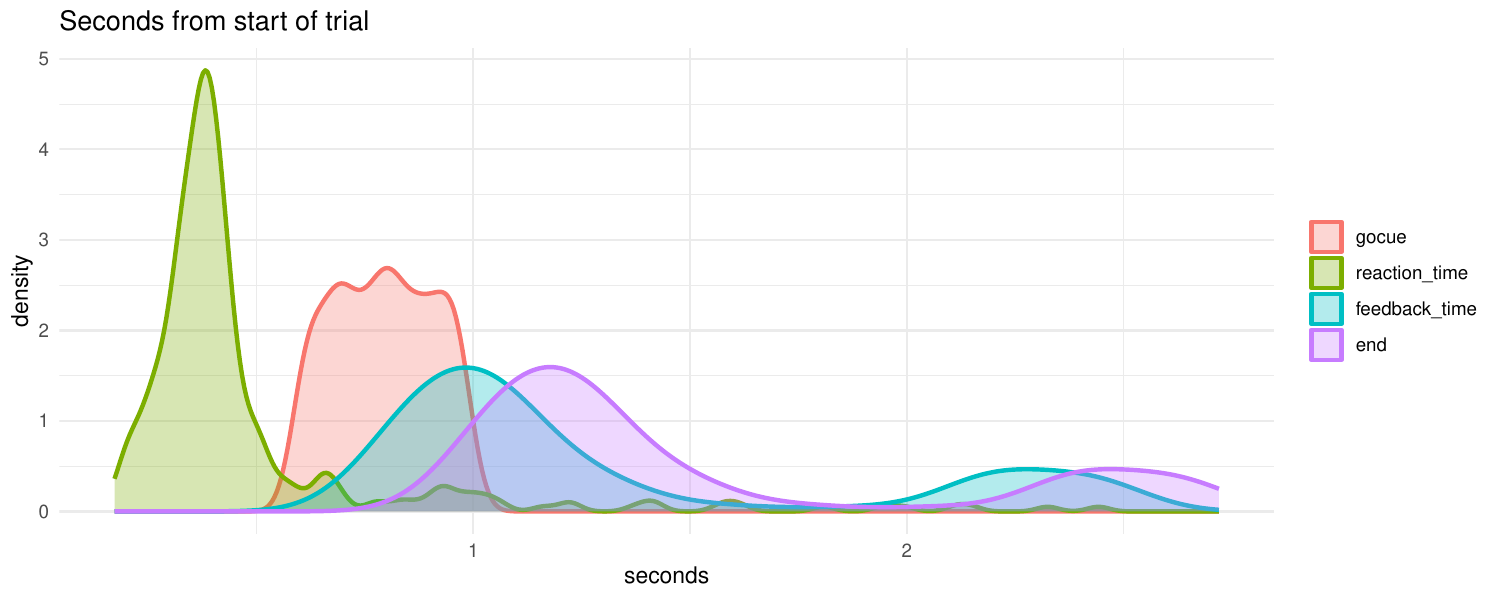}
 \caption{A density plot illustrating the time-distribution of the salient moments across trials.}
 \label{fig:eda_salient}
\end{center}
\end{figure}

The challenges that characterize each trial can have various difficulty levels, implied by the different levels of contrast, as we have outlined in Section \ref{sec:data}. The feedback is a binary variable which is equal to $1$ if the mouse receives water as a reward or $0$ otherwise. Figure \ref{fig:eda_difficulty_feedback_table} empirically summarizes the outcomes out of the $340$ trials, illustrating how positive feedback was more common for easier challenges, as it would be expected.
\begin{figure}[htbp!]
\begin{center}
 \includegraphics[width=0.8\textwidth]{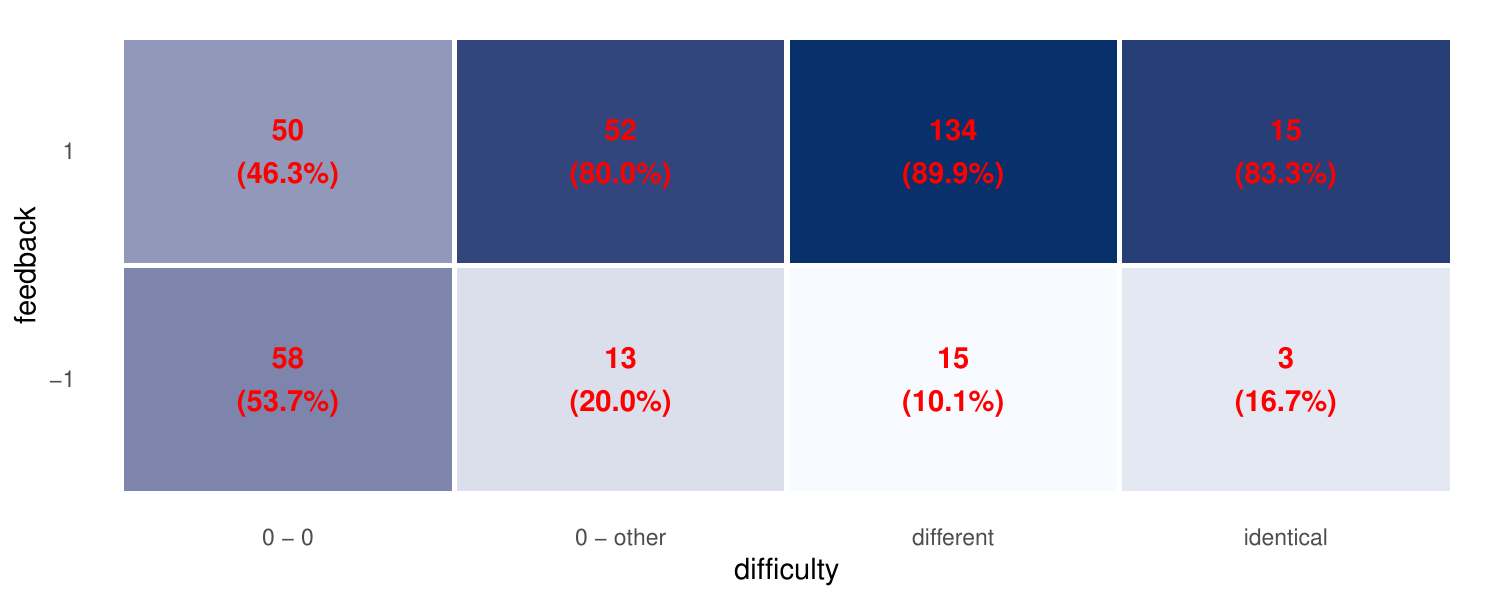}
 \caption{Frequency (percentage) table of trial difficulties and feedback types. The darker color indicates better performance for each difficulty level, independently. This illustrates that the mouse is generally successful, with the exception of the case \texttt{0-0}.}
 \label{fig:eda_difficulty_feedback_table}
\end{center}
\end{figure}

\subsection{Computational aspects}\label{sec:application_mcmc}

Markov chain Monte Carlo sampling was run for a grand total of $160{,}000$ iterations. During the first $150{,}000$ iterations, the proposal variances were individually tuned to achieve an acceptance rate between $20\%$ and $40\%$ for each individual model parameter. After this initial phase, the proposal variances were fixed for the rest of the sampling. The following $5{,}000$ samples were discarded as additional burn-in period. The last $5000$ iterations were used to create the final sample by retaining one every $20$-th observation. 

This procedure effectively returned a final approximate posterior sample of $250$ observations, for each model parameter. We found that this sample size permitted an accurate and nuanced representation of the posterior properties of the model, without overburdening our calculations with a large number of parameters and samples.
Convergence was confirmed using standard diagnostics tools, such as checking the trace plots of the Markov chains and the Raftery-Lewis~\citep{raftery1992practical} diagnostic to estimate required sample sizes. 

Regarding hyperparameters, we set the variance for the intercepts as $\eta_a = 0.1$ and $\eta_w = 0.0001$. For the priors of network parameters, we set $\zeta_a = \omega_a = 0.00005$ and $\zeta_w = \omega_w = 0.005$. This prior formulation allows $\mu$ to reflect changes in baseline activity across trials, while limiting its movements within each trial. By contrast, we let the latent positions vary only a little across trials, but give them more freedom to move within each trial. 
These choices are motivated by the fact that the trial is structured around salient events (stimulus onset, go cue, reaction time, feedback). 
Across these stages, patterns of \emph{inter-regional information flow} (or directed functional interactions) may change \citep{seth2015granger}, and this should be reflected in time-varying network-induced effects. In our model, such stage-dependent reconfiguration is captured by changes in the latent positions (and hence in $\bS_{m,t}$), while the baseline activity levels $\mu$ remain approximately constant during a particular trial.
Changes in the network structure (i.e., directed functional interaction) in the same session remain mostly unchanged across trials, but might evolve over a longer time scope due to neural plasticity.  
Thus, across trials, the network effects are maintained constant to avoid overparametrizing the model and to ensure more accurate estimates of the parameters.
By contrast, we allow the baseline activity level to change across trials. Variation in baseline activity may capture changes of patterns in the mouse's behavior over time, for example due to fatigue or changes in engagement. 
This complementary setting, whereby the baseline activity and network interactions engage with different aspects of the time dynamics, allows us to disentangle the contributions to the temporal dynamics that are given by one model component or the other.
This facilitates inference on the model parameters and interpretability of the results.

\subsection{Results}\label{sec:application_results}

In order to obtain a preliminary overview of the results, we aggregate the point estimates of all model parameters over time. Figure~\ref{fig:res_overview} shows the smoothed evolution across trials of the baseline parameters $\bmu$ (left panel) and the network-corrected log-rates $\log\blambda$ (right panel), plotted on the same scale to highlight how the activity of each brain area is affected by the inferred network interactions. Baseline activity varies substantially across regions: the mediodorsal thalamus (\texttt{MD}) is the most consistently active region, while the anterior cingulate area (\texttt{ACA}) and secondary motor cortex (\texttt{MOs}) are markedly less active. The cross-trial dynamics are modest in scale but well pronounced. Most regions show a slow upward drift, while \texttt{ACA} and \texttt{MOs} drift downward. Because the link function is logarithmic, posterior uncertainty is relatively wider for small values of $\mu$, and readings for low-activity regions should be interpreted with greater caution.

\begin{figure}[htbp!]
\begin{center}
 \includegraphics[width=0.49\textwidth]{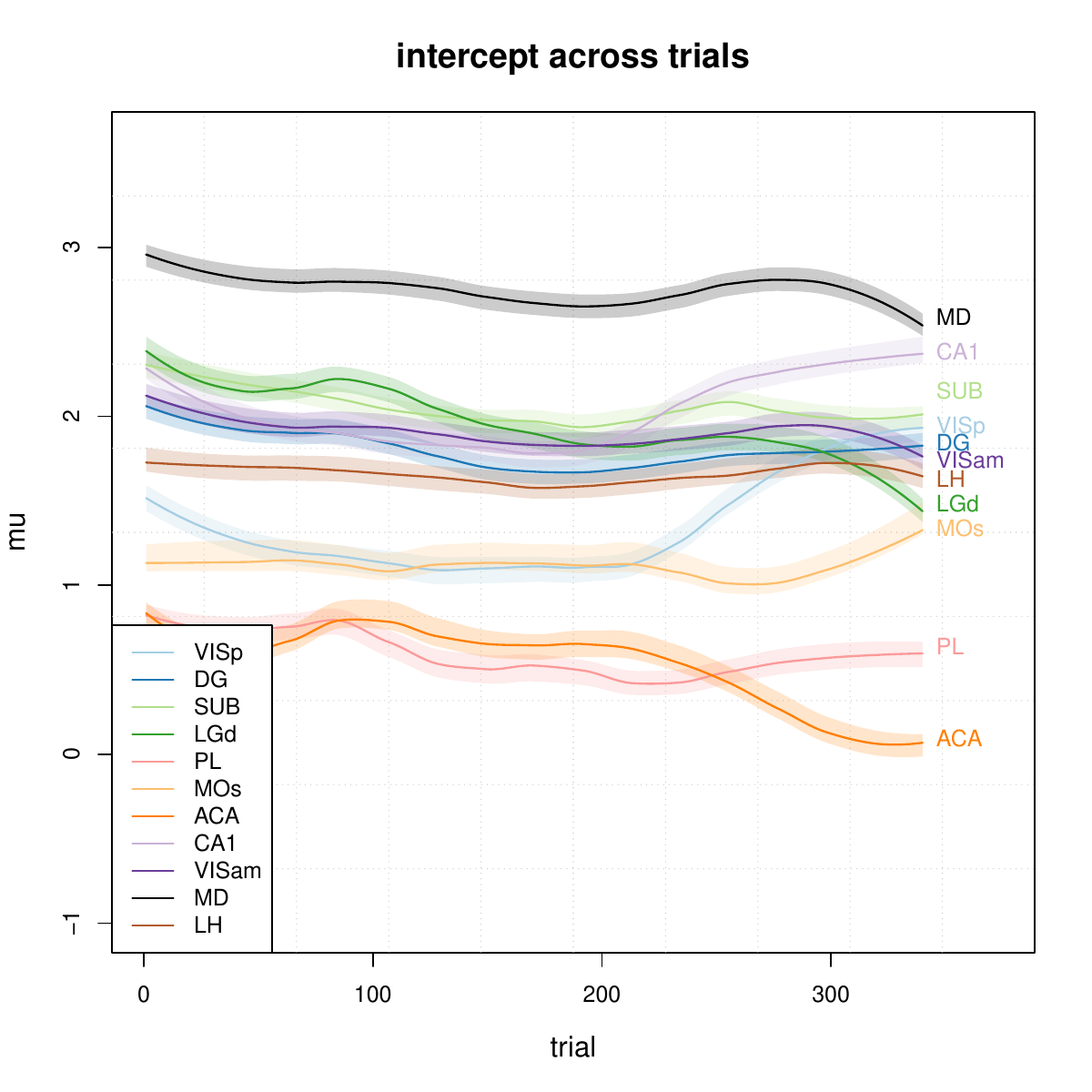}
 \includegraphics[width=0.49\textwidth]{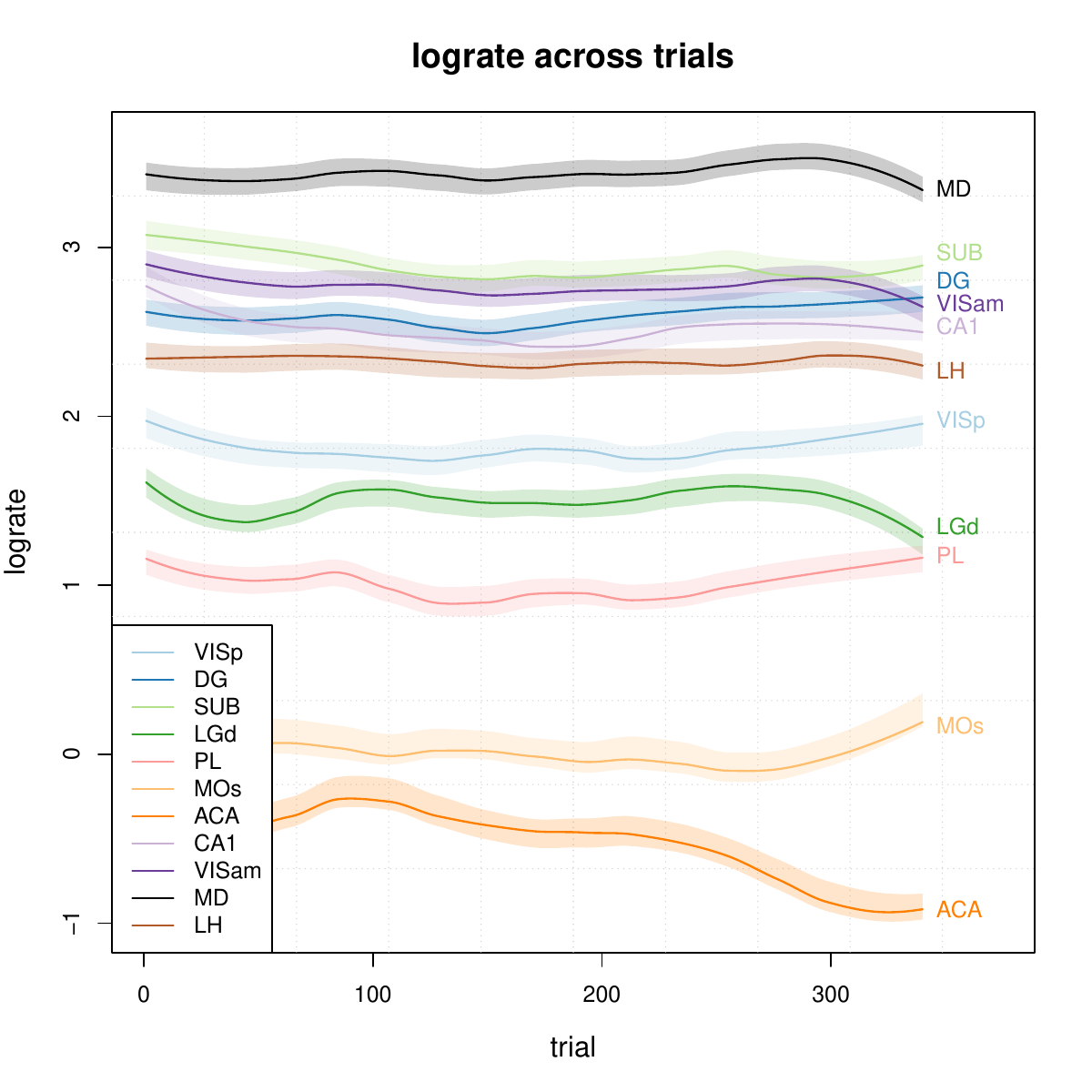}
 \caption{Loess smoothed intercept terms (left panel) and log-rates (right panel) to show their time dynamics across the $340$ trials. The log-rates exhibit more heterogeneity, signaling that the network interactions can amplify or reduce the activity of the brain areas in different ways. Loess smoothed credible intervals are added to represent uncertainty.}
 \label{fig:res_overview}
\end{center}
\end{figure}

%The shapes of the posterior distributions of the positions appear non-Gaussian, with most posterior clouds displaying non-spherical shapes and, to some extent, multi-modality. 
A central contribution of our model is a visual representation of the directed functional architecture of the recorded network. Figure~\ref{fig:res_overview_space} summarizes the inferred latent positions, where each position is averaged across all iterations, trials, and time after Procrustes' rotation. The geometry encodes the network interactions through the dot product $\bz_i \cdot \bz_j$: two regions in the same angular sector excite one another, two regions in opposite sectors inhibit one another, and orthogonal regions interact only weakly, while the magnitude $\|\bz_i\|$ controls how strongly each region participates.

\begin{figure}[htbp!]
\begin{center}
 \includegraphics[width=0.7\textwidth]{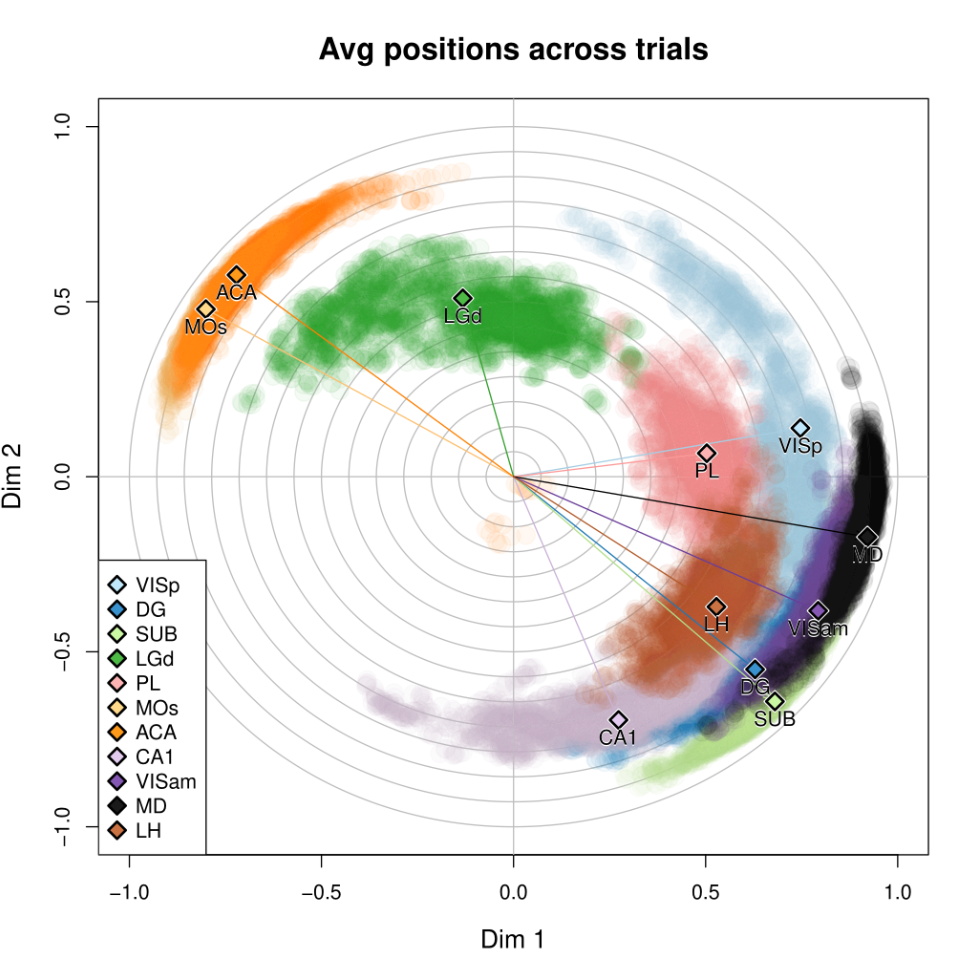}
 \caption{Latent space that arises after Procrustes' rotations illustrating the sampled positions across all brain areas, times, trials and iterations. The colors correspond to the different brain areas. The squares indicate the positions for each brain area averaged out across iterations, times and trials.}
 \label{fig:res_overview_space}
\end{center}
\end{figure}

Two groups of nodes show strong positive within-group interactions: \texttt{MD}, \texttt{CA1}, and the primary visual cortex (\texttt{VISp}); and a second group comprising the dentate gyrus (\texttt{DG}), the lateral hypothalamus (\texttt{LH}), the subiculum (\texttt{SUB}), the higher visual area \texttt{VISam}, and the prelimbic cortex (\texttt{PL}).
The two groups generally compete, in the sense that higher activity in one group reduces activity in the other and vice versa. The dorsal lateral geniculate nucleus (\texttt{LGd}) sits opposite to the second group. This separation is consistent with \texttt{LGd}'s role as the exclusive thalamic relay of retinal input to \texttt{VISp} \citep{kerschensteiner2017organization, guido2018development}, whose trial-by-trial variability is dominated by retinal drive rather than by lateral coupling with the other recorded sites. \texttt{ACA} and \texttt{MOs} sit in strong opposition to most other recorded regions, consistent with their established roles as top-down modulators whose bulk variability is shared with the cortical sites they project to \citep{norman2021chemogenetic, yang2021secondary, zatka2021sensory}.

\subsection{Individual network effects across trials}
We now isolate the role of network interactions and quantify how much they amplify or attenuate each region's activity. We define $\psi_{m,t,i} \equiv \exp\{\log(\lambda_{m,t,i}) - \mu_{m,t,i}\}$ as the network effects, i.e., the multiplicative factor by which the Poisson intensity is modulated as a consequence of network interactions. 
In the absence of network interactions, $\mu_{m,t,i} = \log(\lambda_{m,t,i})$ and $\psi_{m,t,i} = 1$. If the net network effect on node $i$ is positive (resp.\ negative), then $\psi_{m,t,i} > 1$ (resp.\ $\psi_{m,t,i} < 1$). Equivalently, Eq.~\eqref{eq:model_matrix_notation_2} gives $\gamma_{m,t,i} = \log\lambda_{m,t,i}$, the network-adjusted log-rate, so $\psi_{m,t,i} = \exp(\gamma_{m,t,i} - \mu_{m,t,i})$ is the multiplicative correction by which the baseline rate $\exp(\mu_{m,t,i})$ is rescaled to the network-adjusted rate $\lambda_{m,t,i}$. The quantity $\psi$ admits an alpha-centrality interpretation that it represents the local centrality contribution of node $i$ within the directed influence network defined by $\bS_{m,t}$.

The left panel of Figure~\ref{fig:ne_intercept_vs_logrates} shows the time-averaged baseline activity $\mu$ ($x$-axis) against the time-averaged network-corrected log-rate $\log\lambda$ ($y$-axis), each averaged across all trials. Points above the $y=x$ diagonal correspond to $\psi>1$, where the network has a synergistic effect on activity, whereas points below correspond to $\psi<1$, where the network effect suppresses activity. Most regions sit above the diagonal indicating synergistic network effects, except for  \texttt{ACA}, \texttt{MOs}, and \texttt{LGd}.
 
\begin{figure}[htbp!]
\begin{center}
 \includegraphics[width=0.49\textwidth]{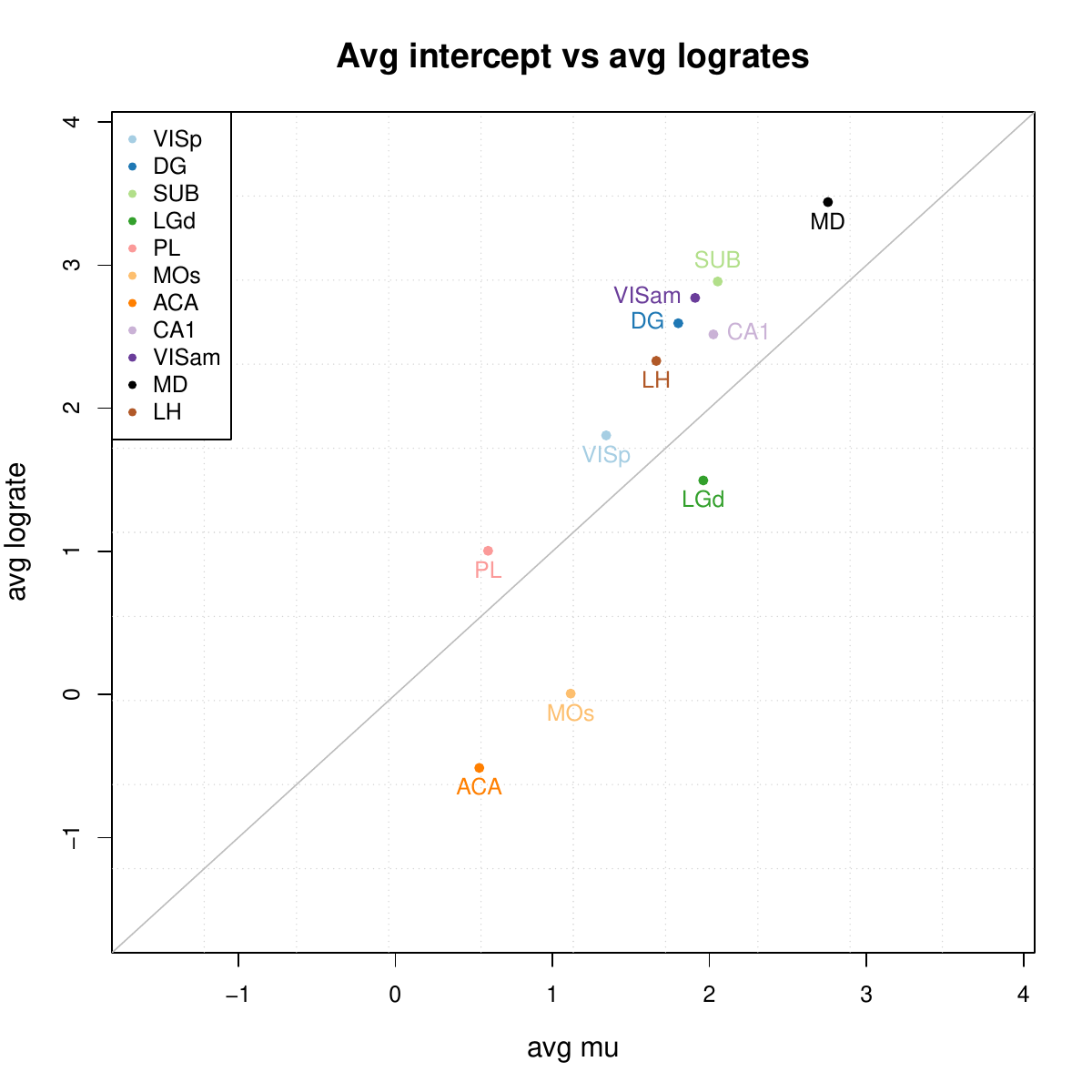}
 \includegraphics[width=0.49\textwidth]{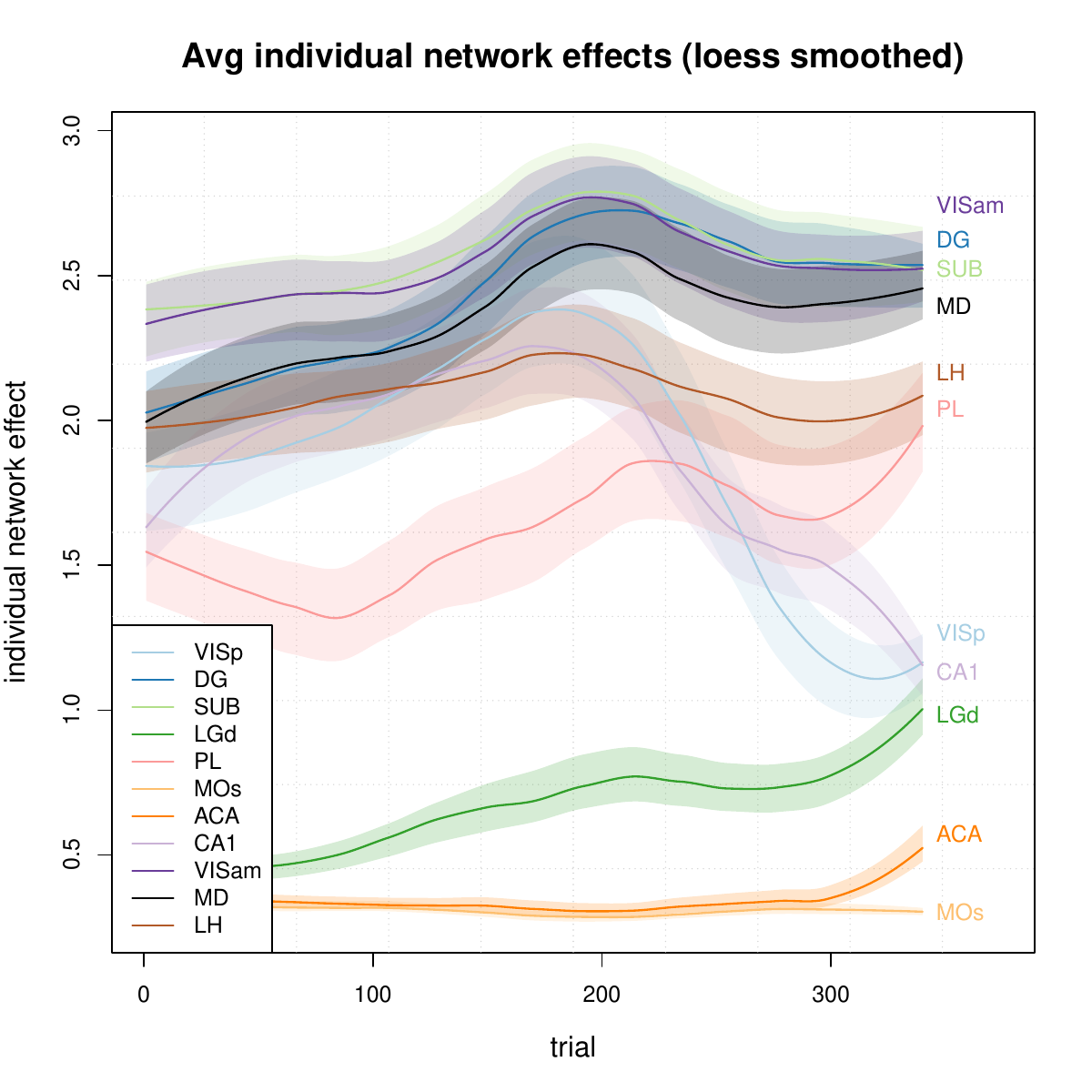}
 \caption{Left panel: contrast between the baseline activity and the activity after the network correction is applied as per Eq.~\eqref{eq:model_matrix_notation_2}. For most brain areas, the network effect has a synergistic role because it increases the activity of the brain areas. Right panel: spillover of activity averaged out over the duration of each trial. The network interactions tend to become weaker as the experiment progresses.}
 \label{fig:ne_intercept_vs_logrates}
\end{center}
\end{figure} 

We next examine how network participation evolves across trials for each region. The right panel of Figure~\ref{fig:ne_intercept_vs_logrates} shows $\psi$ trajectories, smoothed for visual clarity, with one line per region.  Most regions trend downward in the second half of the trials, with the steepest decreases in \texttt{VISp} and \texttt{CA1}. 
\texttt{LGd} and \texttt{ACA} show the opposite trajectory, where their $\psi$ values increase steadily throughout the session. 
These changes show a gradual weakening of the network's multiplicative correction across the session, which draws each region's $\psi$ toward $1$ from its respective side of the diagonal. This is consistent with the within-session engagement decline documented by \citet{steinmetz2019distributed} in the source dataset (Miss-trial streaks toward the end of the session), if disengagement is accompanied by reduced cross-region coupling.
% It should be noted that the prior specification imposes a strong penalization over the initial conditions of each trial, in that the starting positions of the nodes are essentially the same at each trial. However, once the trial starts, the prior lets the nodes can move more freely. The plot shows the average value of the network effects during the duration of each trial: this is not directly affected by the prior. The lines have been smoothed to facilitate interpretation, and the general trend appears to be downwards.

\subsection{Global measures of network effects within trials}

We propose another perspective on the network dynamics, but we now focus on an aggregation of the trajectories within trials, to assess the time evolution of the network effects in relation to the critical moments of the trials. 
Studying individual within-trial network effects is challenging because of the model's multi-dimensionality: (1) the interaction matrix is time-varying within each trial, and (2) the latent trajectories may differ across trials even when they share the same starting point, making direct visualization difficult. We therefore propose a global measure of network interactions, starting from the matrix $\bS_{m,t}$ defined in Eq.~\eqref{eq:def_S_matrix}. The $(i,j)$th entry of this matrix records the directed network influence from node $j$ on node $i$ at time $t$ of trial $m$. We take the spectral radius of $\bS_{m,t}$ as a global summary of how strongly nodes can affect one another at that moment. This is directly related to alpha centrality and Proposition~\ref{prop:gamma_existence}: the spectral radius of $\bS_{m,t}$ measures the magnitude of network spillover between nodes, and the statistic therefore complements the local $\psi$ statistic introduced in the preceding subsection by quantifying network coordination at the whole-network level.

We illustrate the change of this spectral measure across trials using box plots, aligned to three behaviourally salient moments: stimulus onset, reaction time, and feedback time. Figure~\ref{fig:spectral_within_shifted_stimulus} shows the spectral measure aligned to stimulus onset, with box plots summarizing the variability across trials at each within-trial time point. Network coordination rises sharply within the first few hundred milliseconds after stimulus onset, most notably for the \texttt{identical} trial setting. In these trials discrimination is non-trivial, plausibly increasing the demand for cross-region evidence integration relative to easier trials. We cannot causally attribute this difficulty modulation to evidence integration in our setting, but the pattern is consistent with the broader finding that effective connectivity from sensory to higher-order cortical regions scales with perceptual difficulty in human fMRI \citep{lamichhane2015perceptual}, and that ambiguous-evidence integration in mice is a distributed cortical computation rather than a localized one \citep{pinto2022multiple}.

\begin{figure}[htbp!]
\begin{center}
 \includegraphics[width=0.99\textwidth]{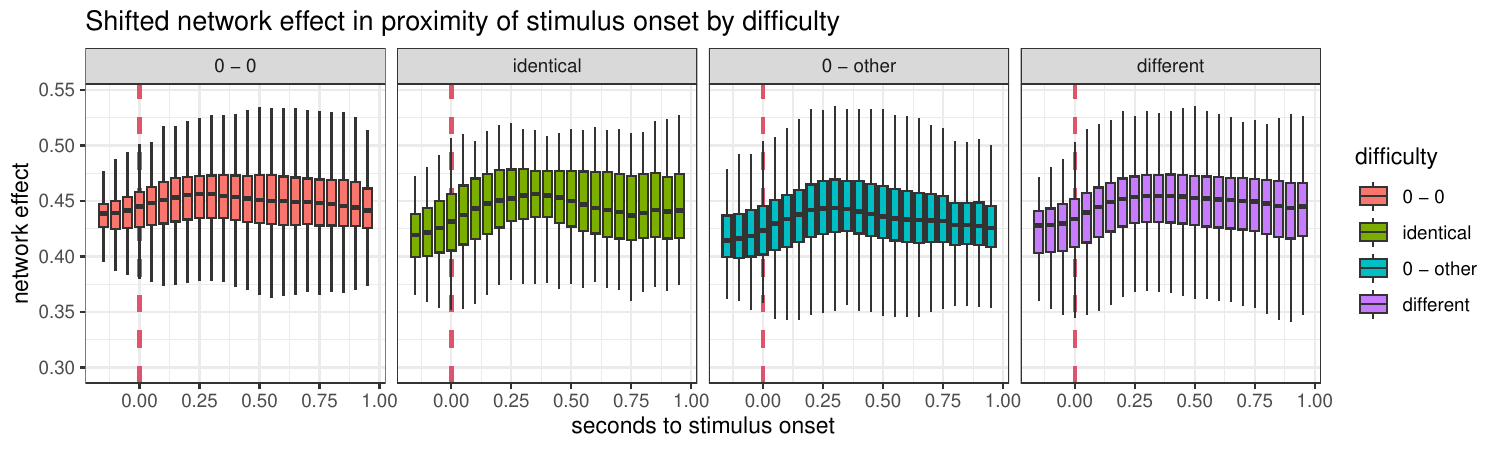}
 \caption{Network effect measured by the spectral radius of $\bS_{m,t}$. The timing of trials IS aligned according to the stimulus onset. The vertical line shows the moment when the images are shown to the mouse. These plots illustrate that, regardless of trial type, the network interactions are minimal before the images are shown to the mouse. After the stimulus onset, the network effect increases for all trial types.}
 \label{fig:spectral_within_shifted_stimulus}
\end{center}
\end{figure}

Figure~\ref{fig:spectral_within_shifted_reaction} shows the spectral measure aligned to the reaction time. 
Here, we see a strong increase in the network effect in correspondence of the mouse's reaction, across all difficulties, to a similar extent. 
This is consistent with previous observations in the same task family that movement onset is accompanied by broadly distributed, transiently coordinated activity across cortical and subcortical regions \citep{steinmetz2019distributed,coen2023mouse}. The peri-reaction-time spike in our spectral measure thus likely reflects the network-level signature of this transient, brain-wide motor coordination.
The results do not show much difference between trial setups, indicating similar behavior of the mouse across experimental conditions.

\begin{figure}[htbp!]
\begin{center}
 \includegraphics[width=0.99\textwidth]{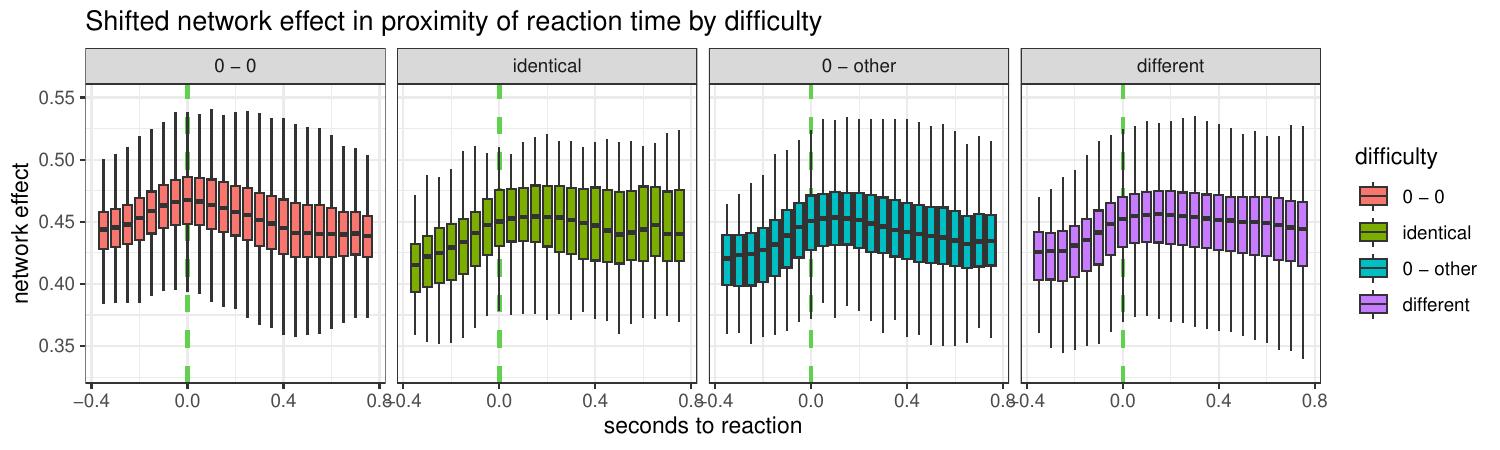}
 \caption{Network effect measured by the spectral radius of $\bS_{m,t}$. The timing of trials IS aligned according to the moment when the mouse reacts to the stimulus. The vertical line indicates this moment of the trials. These plots illustrate that the network interactions increase sharply for all trial settings, leading up to and immediately after the reaction of the mouse.}
 \label{fig:spectral_within_shifted_reaction}
\end{center}
\end{figure}

\begin{figure}[htbp!]
\begin{center}
 \includegraphics[width=0.99\textwidth]{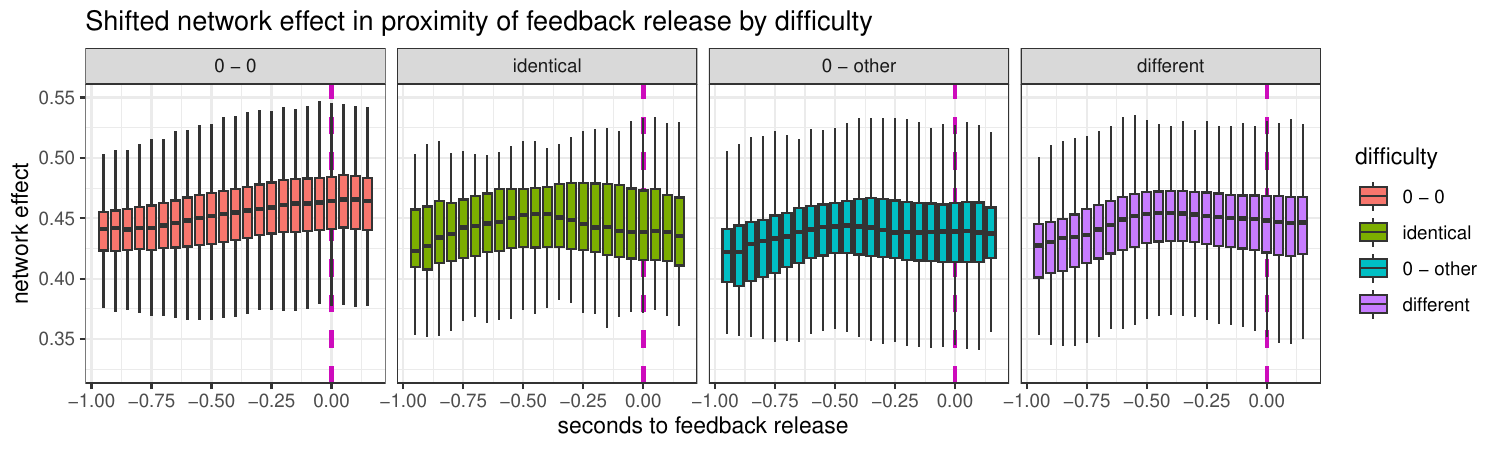}
 \caption{Network effect measured by the spectral radius of $\bS_{m,t}$. The trials are aligned to the instant when the feedback is delivered. The vertical line indicates this moment. The network effect does not show a consistent trajectory across difficulty levels, indicating that the mouse may display diverse behavioral patterns in correspondence to this timing. As shown in Figure \ref{fig:eda_salient}, the feedback timing can vary a lot across trials, thus it does not facilitate the isolation of its effect.}
 \label{fig:spectral_within_shifted_feedback}
\end{center}
\end{figure}

In this case, the time dynamics of the network effect are not consistent across difficulty levels. In some cases, the network effect peaks before feedback delivery; in other cases, it peaks afterward. This can suggest more heterogeneity which may be associated to other aspects of each trial. Since there is no apparent difference between the various difficulties, this may suggest that the animal does not react strongly to different outcomes, and perhaps remains indifferent to the feedback type.

\subsection{Latent spaces at salient moments}

We now leverage the within-trial dynamics of the latent positions to extract the model's view of the network at the most salient moments of the task. We shift the time dimension of each trial to match the inferred latent spaces at five instants: trial start, stimulus onset, reaction time, feedback release, and trial end. Each latent position in Figure~\ref{fig:latent_spaces_critical} is averaged across all iterations and across all trials in which the corresponding salient moment is reached, with arrows indicating the trajectory of each region from one salient moment to the next.

\begin{figure}[htbp!]
\begin{center}
 \includegraphics[width=0.7\textwidth]{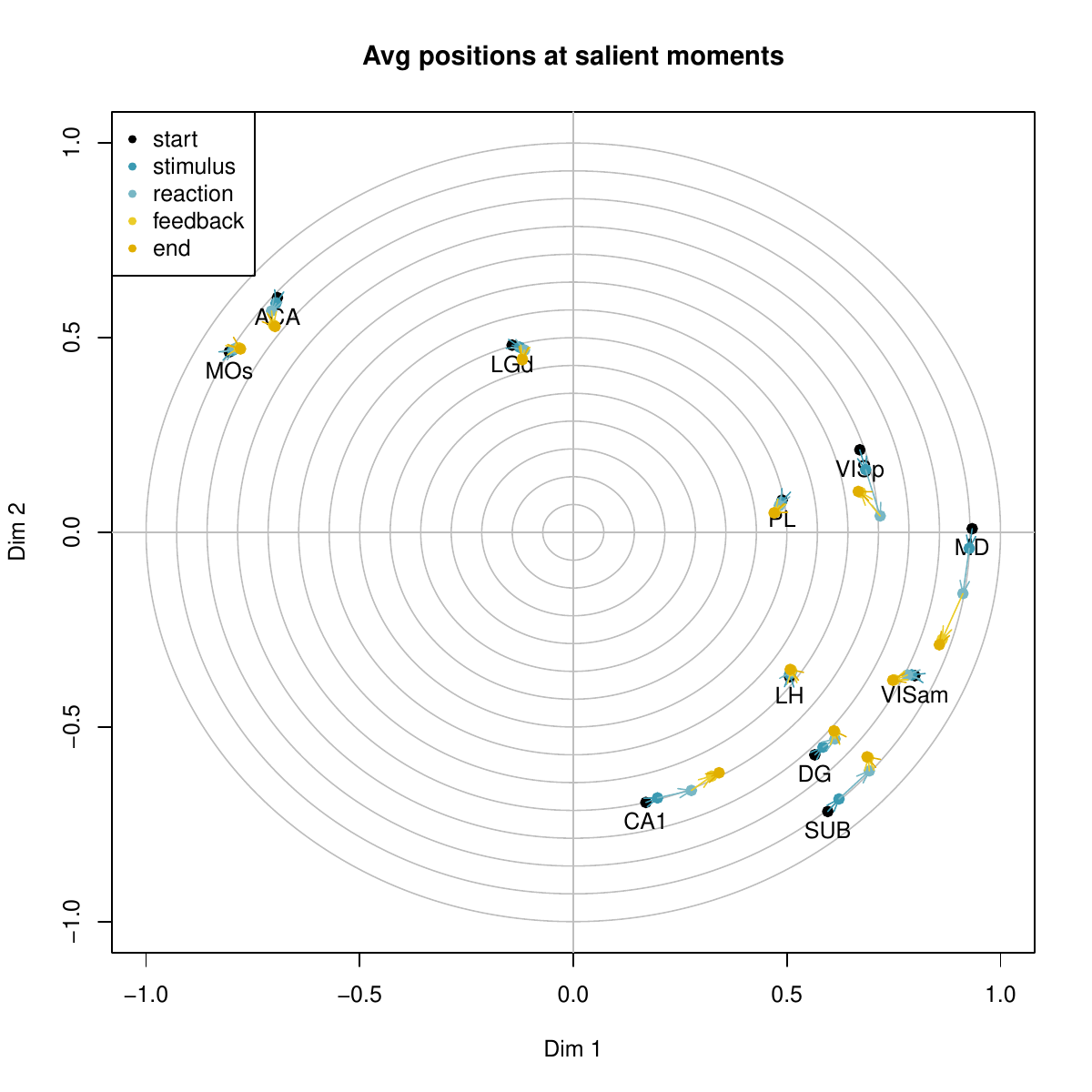}
 \caption{Average latent positions of the brain areas extracted at five salient instants of each trial. The plot illustrates the hidden framework which explains the network interactions via latent space snapshots at the critical stages of the trials. In correspondence to the reaction time, the nodes tend to tighten the radius of the latent space, suggesting a stronger (but temporary) synergy during this phase. The level of participation of nodes in this coordinated movement seems rather heterogeneous.}
 \label{fig:latent_spaces_critical}
\end{center}
\end{figure}

Each of the latent positions in the plot is averaged out across all iterations and across all trials where the corresponding salient moment is reached. 
The plot shows that the nodes remain relatively stable, but a more noticeable movement spike is generally observed before the reaction time. During this time, nodes seem to generally converge towards each other, indicating a radial contraction on the latent space. This is apparent for the nodes \texttt{SUB}, \texttt{VISp}, \texttt{CA1} and \texttt{MD}. The movement of these brain regions tends to align them better with the direction of several other nodes, which is consistent with the increase in activity just before the reaction time.

%!TEX root = ../sn-article.tex

\section{Discussion}\label{sec:discussion}

We introduced a new methodology to analyze multivariate time series and proposed an application to spike-train data. A central and novel aspect of our proposed framework is that we infer network interactions between the time series, and characterize these using a latent space approach. This approach leads to new model-based visualizations and perspectives on the interactions between brain areas, breaking down complex patterns and explaining them using latent variables. Critically, the model features a nested hidden Markov temporal structure which is coherent with the study design, and is able to capture the non-linear time dynamics of the experimental conditions.

We performed inference under a Bayesian framework and we demonstrated with simulations that the model parameters can be accurately recovered for large datasets, identifying the hidden network interactions and their effects on the time series. Our application to spike-train data studied how the brain areas of a mouse activate during a sequence of experimental tasks, and it illustrated how the brain areas can synergize or compete during the experiment. As an output of our method, we introduced several model summaries which quantify the network interactions and their effects on the baseline activity levels of the brain regions.

Our work leaves several possible avenues for future research. In connection to the literature on latent variable network models, a natural variant of our approach could be to use a Euclidean distance based latent space, as in the pioneering work of \citet{hoff2002latent}, or other types of latent space geometries.
Alternatively, adaptations of the dynamic stochastic blockmodel~\citep{matias2017statistical} could also lead to interesting results, considering that this can generally scale better with the size of the data, and that categorizations of brain regions can also be available. 

Finally, computational efficiency remains a critical obstacle for latent space models, but also, more in general, for network-based frameworks. As we have shown in this work, our methodology scales with the square of the number of nodes, which makes it impractical for larger datasets. In the context of brain networks, we have bypassed this limitation by considering an aggregation of neurons into brain areas, however, one could also consider faster approaches such as approximate Markov chain Monte Carlo~\citep{rastelli2018computationally} or variational inference~\citep{salter2013variational} to extend the applicability of this new framework.

\backmatter

% \bmhead{Supplementary information}

\bmhead{Acknowledgments}
This publication has emanated from research conducted with the financial support of Taighde \'Eireann – Research Ireland, under Grant number $[23/RC/13506]$ at the Research Ireland Centre - Rinn Artificial Intelligence.

% \section*{Declarations}

% \begin{appendices}

% \section{Section title of first appendix}\label{secA1}

% An appendix contains supplementary information that is not an essential part of the text itself but which may be helpful in providing a more comprehensive understanding of the research problem or it is information that is too cumbersome to be included in the body of the paper.

% \end{appendices}

\bibliography{sn-bibliography}

\appendix
%!TEX root = ../sn-article.tex

\section*{Appendix: Poisson log-Normal projection model}\label{sec:pln_projection}

While numerous frameworks are available to model abundance-type data using latent variables, none of these methods embed a latent space model to characterize the nodes' interactions. For this reason, there is no obvious method that we can use for a direct comparison. Besides the novel proposal of our research article, one may attempt to model abundance data using a two-steps procedure: one can initially fit a generalized Poisson regression model to the counts, and then fit a network model based on the inferred interaction or correlation matrix. We should note that this ad-hoc approach disregards any temporal aspects that may be present in the data, and is proposed as it may be an intuitive approach underpinned by currently available methods.

In this section we employ this two-steps procedure on the real dataset of Section \ref{sec:rda}. We start by aggregating all the counts across all time points, for each individual trial. This leaves us with a matrix of counts for $340$ trials and $11$ brain areas.
We use the function \texttt{PLN} from the \texttt{PLNmodels} package to fit a Poisson log-Normal model to the data. The model introduces an intercept and models the counts directly.

Figure \ref{app:pln_fit} (left panel) shows that the Poisson log-Normal model can fit the data very well.
\begin{figure}[htbp!]
\begin{center}
 \includegraphics[width=0.49\textwidth, page=2]{Figures/pln_fit_aggr.pdf}
 \includegraphics[width=0.49\textwidth, page=1]{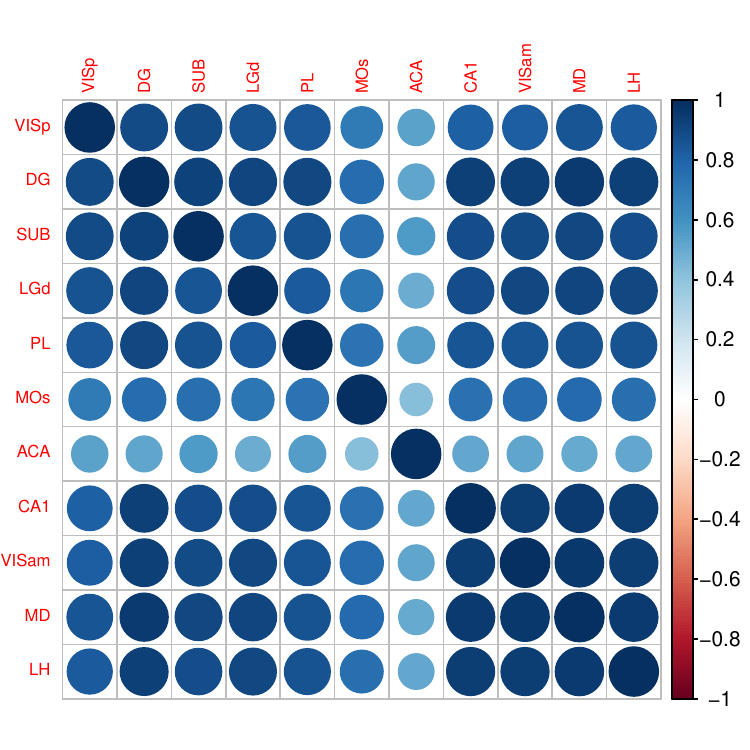}
 \caption{The left panel shows the predicted vs observed counts according to the fitted Poisson log-Normal model. The right panel shows the resulting correlation matrix.}
 \label{app:pln_fit}
\end{center}
\end{figure}
The right panel of the same figure highlights the inferred interactions between the brain areas. On the basis of this correlation matrix, we construct a network adjacency matrix that is neither too sparse nor too dense. We calculate the median of the correlation values, and then create the adjacency matrix by thresholding the correlations at this value. That is, an edge is created between two brain regions if and only if their correlation in the fitted Poisson log-Normal model is greater than the median correlation value. 
This leaves us with an undirected network where approximately half of the possible edges appear, to indicate brain areas whose spike values tend to positively correlate.

We proceed by fitting a $2$-dimensional projection model using the \texttt{R} package \texttt{latentnet}, using default settings. The resulting latent space is shown in Figure \ref{app:pln_lsm}.
\begin{figure}[htbp!]
\begin{center}
 \includegraphics[width=0.7\textwidth]{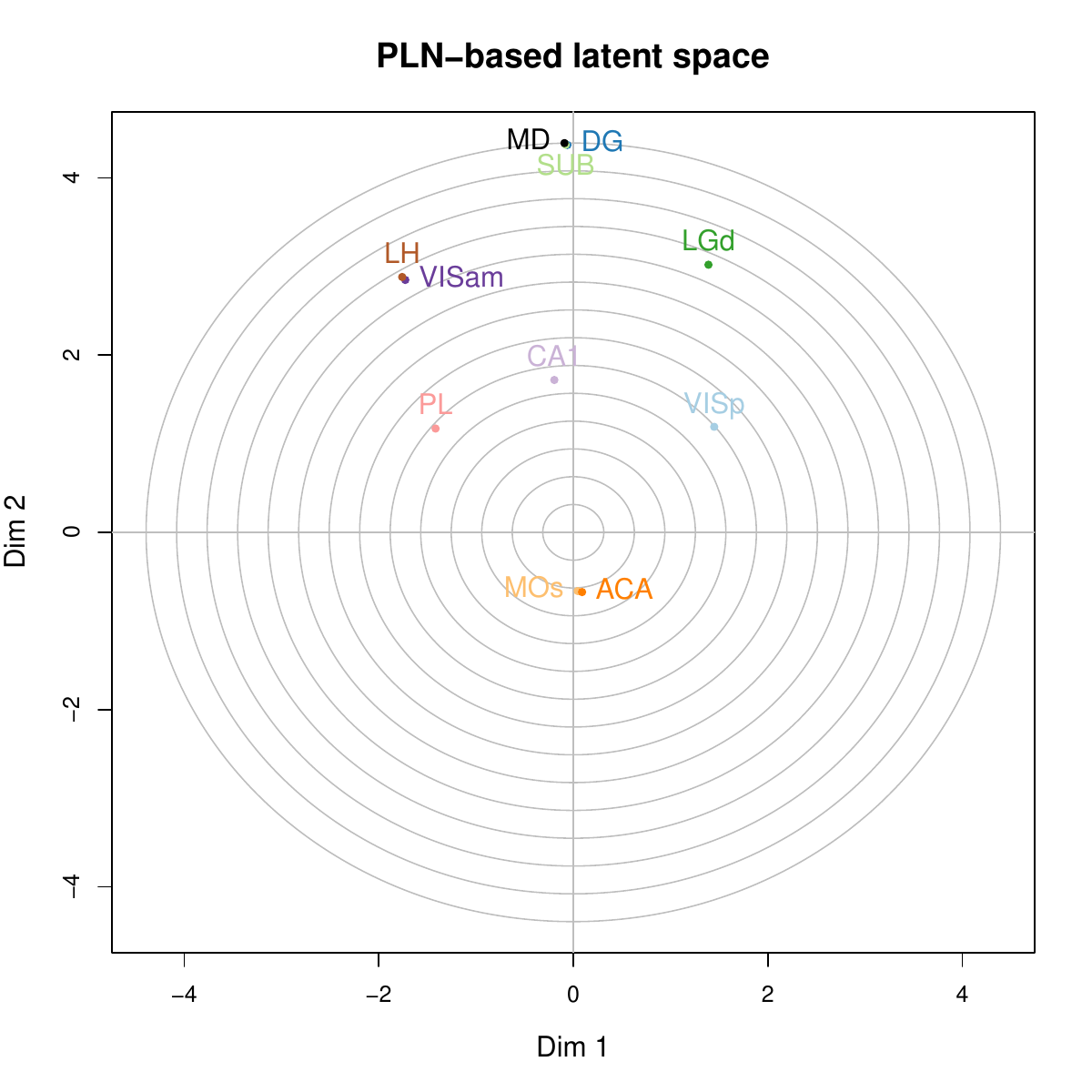}
 \caption{Resulting latent space for a projection model fitted on the correlation matrix of a Poisson log-Normal model on time aggregated data.}
 \label{app:pln_lsm}
\end{center}
\end{figure}
This latent space shares some similarities with that obtained with our proposed framework. The approach captures the separation of regions \texttt{ACA} and \texttt{MOs} (but not \texttt{LGd}) from the rest. The positions of the other brain areas do not seem coherent with those obtained with our proposed framework.

Overall, the two-step procedure based on the Poisson log-Normal and the projection model can capture some features of the data at low computing cost. However this approach does not take into account the nested temporal aspects of this dataset, and does not allow for the in-depth model-based characterization based on the temporal trajectories and network centrality scores. 

We should note that we also tried to fit a Poisson log-Normal model using the \texttt{PLNnetwork} function, which delivers an inferred network instead of a correlation matrix. However, in this case the inferred network did not have any edges, so we dismissed the results.

\end{document}